\documentclass[10pt,a4paper]{article}
\usepackage[top=69pt,bottom=69pt,left=65pt,right=65pt]{geometry}
\usepackage[cmex10]{amsmath}
\usepackage{algorithmic}
\usepackage{algorithm}
\usepackage{subcaption}
\usepackage{amsthm}
\usepackage{graphicx}
\usepackage{accents}
\usepackage{float}
\usepackage{undertilde}
\usepackage{color}
\usepackage{caption}
\usepackage{multicol}

\usepackage{amssymb}
\newcommand{\Real}{\mathbf{R}}
\newcommand{\Integ}{\mathbb{Z}}
\DeclareMathOperator{\argmin}{arg\,min}
\DeclareMathOperator{\argmax}{arg\,max}
\newtheorem{theorem}{Theorem}
\newtheorem{lemma}{Lemma}
\newtheorem{definition}{Definition}
\newtheorem{proposition}{Proposition}
\newtheorem{remark}{Remark}

\makeatletter
\newcommand{\SPAN}{\mathop{\operator@font span}}
\makeatother
\makeatletter
\newcommand{\supp}{\mathop{\operator@font supp}}
\makeatother

\begin{document}

\title{Robust Linear Regression Analysis - A Greedy Approach}

\author{George~Papageorgiou\footnote{geopapag@di.uoa.gr} , Pantelis~Bouboulis\footnote{panbouboulis@gmail.com} and~Sergios~Theodoridis\footnote{stheodor@di.uoa.gr}
}
\maketitle

\begin{abstract}
The task of robust linear estimation in the presence of outliers is of particular importance in signal processing, statistics and machine learning. Although the problem has been stated a few decades ago and solved using classical (considered nowadays) methods, recently it has attracted more attention in the context of sparse modeling, where several notable contributions have been made. In the present manuscript, a new approach is considered in the framework of greedy algorithms. The noise is split into two components: a) the inlier bounded noise and b) the outliers, which are explicitly modeled by employing sparsity arguments. Based on this scheme, a novel efficient algorithm (Greedy Algorithm for Robust Denoising - GARD), is derived. GARD alternates between a least square optimization criterion and an Orthogonal Matching Pursuit (OMP) selection step that identifies the outliers. The case where only outliers are present has been studied separately, where bounds on the \textit{Restricted Isometry Property} guarantee that the recovery of the signal via GARD is exact. Moreover, theoretical results concerning convergence as well as the derivation of error bounds in the case of additional bounded noise are discussed. Finally, we provide extensive simulations, which demonstrate the comparative advantages of the new technique.
\end{abstract}


\section{Introduction}
\label{Intro}
%
%
%
%
The notion of {\em robustness}, i.e., the efficiency of a method to solve a learning task from data, which have been contaminated by large values of noise, has occupied the scientific community for over half a century \cite{dixon1950analysis, grubbs1969procedures}. Regardless the type of the problem, e.g., classification or regression, the goal is to identify observations that have been hit by large values of noise, known as {\it outliers}, and be removed from the training data set. Over the years, many authors have tried to state definitions to identify a data point as an outlier. A few typical characterizations follow:
\begin{itemize}
\item ``An outlier is an observation that deviates so much from other observations as to arouse suspicions that is was generated by a different mechanism" (Hawkins, 1980) \cite{hawkins1980identification}.
\item ``An outlier is an observation which appears to be inconsistent with the remainder of the data set" (Barnet and Lewis, 1994) \cite{barnett1994outliers}.
\item ``An outlier is an observation that lies outside the overall pattern of a distribution" (Moore and McCabe, 1999) \cite{moore1989introduction}.
\item ``An outlier in a set of data is an observation or a point that is considerably dissimilar or inconsistent with the remainder of the data" (Ramasmawy, 2000) \cite{ramaswamy2000efficient}.
\item ``Outliers are those data records that do not follow any pattern in an application" (Chen, 2002) \cite{tang2002enhancing, chen2003modeling, tang2007capabilities}.
\end{itemize}

In this paper, we focus on solutions of the linear regression problem in the presence of outliers. In such tasks, classic estimators, e.g., the Least-Squares, are known to fail \cite{huber19721972}. This problem has been addressed since the 1950's, in \cite{dixon1950analysis, grubbs1969procedures} and actually solved more than two decades later, in \cite{huber19721972, rousseeuw1990unmasking, leroy1987robust}, leading to the development of a new field in Statistics, known as \textit{Robust Statistics}.


The variety of methods that have been developed to handle outliers can be classified into two major directions. The first one includes methods that rely on the use of \textit{diagnostic tools}, where one tries first to delete the outliers and then to fit the ``good" data by Least-Squares. The second direction, i.e., \textit{robust analysis}, includes methods that firstly fit the data, using a rough approximation, and then exploit the original estimation to identify the outliers as those points which possess large residuals. Both approaches have a long history. Methods developed under the Robust statistics framework, consist of combinatorial optimization algorithms like Hampel's Least Median of Squares Regression (LMedS) \cite{rousseeuw2005robust} (p.16), Fischler's and Bolles's Random Sample Consensus (RANSAC) \cite{fischler1981random}, as well as Rousseeuw's Least Trimmed Squares (LTS) \cite{maronna2006robust, rousseeuw2005robust}. Combinatorial optimization methods seemed to perform well at that time, although they were never adopted by the community. Nowadays, where the size of the training data set can be very large such techniques are prohibited. In contrast, the desire for lower complexity efficient algorithms has constantly been rising. One of the pioneering research works at that time, was the development of Huber's M-est \cite{huber1981wiley, maronna2006robust, huber19721972}, a method that belongs to the category of robust analysis. M-est provides good estimates, without a heavy computational cost, using robust functions of the residual norm (instead of the square function), in order to penalize large values of the residual.

The development of methods in the spirit of robust analysis, owes a lot to the emergence of \textit{sparse} optimization methods, during the past decade. Sparsity-aware learning and related optimization techniques have been at the forefront of the research in signal processing, encompassing a wide range of topics, such as compressed sensing, denoising and signal approximation techniques  \cite{chen1998atomic, candes2006stable, tibshirani1996regression, lustig2007sparse, candes2006robust,theodoridis2015machine_learning}.
There are two major paths, towards modeling sparse vectors/signals. The first one, focuses on minimizing the $\ell_0$ (pseudo)-norm of a vector which equals the number of non-zero coordinates of a vector (this is a non-convex function), whereas the second one employs the $\ell_1$ norm, the closest convex relaxation to the $\ell_0$ (pseudo)-norm to regularize the Least-Squares cost function. Both methods have been shown to generate sparse solutions.

The family of algorithms that have been developed to address problems involving the $\ell_0$ (pseudo)-norm, comprises \textit{greedy} algorithms, which have been shown to provide the solution of the related minimization task, under certain reasonable assumptions, \cite{tropp2007signal, davis1997adaptive, mallat2008wavelet, needell2010signal, needell2009cosamp}. Even though, in general, this is an NP-hard problem, it has been shown that such methods can efficiently recover the solution in polynomial time. On the other hand, the family of algorithms developed around the methods that employ the $\ell_1$ norm, embraces convex optimization, which provide a broader set of tools and stronger guarantees for convergence \cite{mallat2008wavelet, chen1998atomic, candes2005decoding, candes2006stable, tibshirani2013lasso,theodoridis2015machine_learning}.


In the present paper, the robust linear regression problem is approached via a celebrated greedy algorithm, the so called Orthogonal Matching Pursuit (OMP). The main idea is to split the noise into two separate components; one that corresponds to the inlier bounded noise and the other to the outliers. Since the outlier part is not present in all samples, sparse modeling arguments are mobilized to model the outlier noise component. This concept has been also employed in \cite{jin2010algorithms,mateos2012robust,mitra2010robust}. The novelty of our approach lies in the different modeling compared to already established works, by treating the task in terms of the $\ell_0$ minimization via greedy algorithmic concepts. A new algorithm has been derived which exhibits  notable performance gains, both in terms of computational resources as well as in terms of quality of the recovered results. Moreover, theoretical results concerning the power of the method to recover the outiers as well as performance error bound are derived.

The paper is organized as follows: In Section \ref{prel&relwork}, a brief overview of the various methods that address the robust linear regression task is given. Section \ref{sec:GARD} presents in details the proposed algorithmic scheme (GARD). The main theoretical results concerning GARD, including convergence, recovery of the support, recovery error, e.t.c. are included in section \ref{sec:theoretic}. To validate the proposed method, section \ref{sec:exper} includes several experiments, which compare GARD with other existing techniques. It is shown that GARD offers improved recovery at a reduced computational requirements. Finally, Section \ref{SEC:CONCL} contains some concluding remarks.

\textbf{Notations:} Throughout this work, capital letters are employed to denote sets, e.g., $S$, where $S^{c}$ and $|S|$ denote the complement and the cardinality of the $S$ respectively. The set of integer numbers between $1$ and $n$, i.e., $\{1,2,\dots,n\}$, will be denoted as $1\textrm{---}n$. Bold capital letters denote matrices, e.g., $\mathbf{X}$, while bold lowercase letters are reserved for vectors, e.g., $\boldsymbol\theta$. The symbol $\cdot^T$ denotes the transpose of the respective matrix/vector. The $i-$th column of  matrix $\mathbf{X}$ is denoted by $\boldsymbol{x}_i$ and the $i-$th element of vector $\boldsymbol\theta$ is denoted by $\theta_i$. The matrix $\mathbf{X}_S$ is the matrix $\mathbf{X}$ restricted over the set $S$, i.e., the matrix which
comprises of the columns of  $\mathbf{X}$, whose indices belong to the ordered index set $S=\{j_1< \dots< j_s\}$.
Moreover, the identity matrix of dimension $n$ will be denoted as $\mathbf{I}_n$, the zero matrix of dimension $n\times m$, as $\mathbf{O}_{n\times m},$ the vector of zero elements of appropriate dimension as $\mathbf{0}$ and columns of matrix $\mathbf{I}_n$ restricted over the set $S$, as $\mathbf{I}_{S}$.
If $\boldsymbol{v} \in \Real^n$ is an $s$-sparse vector over the support set $S\subset 1\textrm{---}n$, with $|S|=s$, then we denote as $[\boldsymbol{v}]_S\in \Real^s$, or $\boldsymbol{v}_S$ for sort,  the vector which contains only the $s$ non-zero entries of $\boldsymbol{v}$, i.e., $\boldsymbol{v}_S = \mathbf{I}_S^T\boldsymbol{v}$. For example, if $\boldsymbol{v}=(5,0,0,2,0)^T$, then $\boldsymbol{v}_{\{1,4\}} = (5,2)^T$. Moreover, one can easily show that $\mathbf{I}_S \boldsymbol{v}_S=\boldsymbol{v}$. Finally, we use a linear tensor defined as $F_{S'}(\boldsymbol{\alpha}) = \mathbf{I}_{S'}\mathbf{I}_{S'}^T\boldsymbol{\alpha}$, for any vector $\boldsymbol{\alpha}\in\Real^n$ over the index set $S'\subseteq 1\textrm{---}n$, which has identical coordinates with $\boldsymbol{\alpha}$ in all indices of $S'$ and zero everywhere else. Finally, it is obvious  that for the sparse vector $\boldsymbol{v}$ with support  set $S$, $F_S(\boldsymbol{v}) = \mathbf{I}_S\mathbf{I}_S^T\boldsymbol{v}=\boldsymbol{v}.$

\section{Preliminaries and related work}
\label{prel&relwork}

In a typical linear regression task, we are interested in estimating the linear relation between two variables, $\boldsymbol{x}\in\Real^m$ and $y\in\Real$, i.e., $y = \boldsymbol{x}^T\boldsymbol\theta$,  when several noisy instances, i.e., $\{(y_i, \boldsymbol{x}_i),\ i=1,...,n$\},  are known. In this context, we usually adopt the following (regression) modeling
\begin{equation}
y_i = \boldsymbol{x}_i^T\boldsymbol\theta_0+e_i,\ i=1,...,n,
\label{eq:probdef}
\end{equation}
where $e_i$ is some observation noise.
Hence, our goal is to estimate $\boldsymbol\theta_0\in \Real^m$ from the given training dataset of $n$ observations. In matrix notation, Eq. (\ref{eq:probdef}) can be rewritten as follows:
\begin{equation}
\boldsymbol{y} = \mathbf{X}\boldsymbol\theta_0+\boldsymbol{e},
\label{eq:probdefmat}
\end{equation}
where $\boldsymbol{y}=(y_1,y_2,...,y_n)^T,\ \boldsymbol{e}=(e_1,e_2,...,e_n)^T$ and  $\mathbf{X}=[\boldsymbol{x}_1,\boldsymbol{x}_2,...,\boldsymbol{x}_n]^T \in \Real^{n\times m}$.

As it is common in regression analysis, we consider that the number of observations exceeds the number of unknowns, i.e., $n>m$. In order to seek for a solution, we should assume that $\mathbf{X}$ is a full rank matrix, i.e. $rank(\mathbf{X})=m$.
If the noise is i.i.d Gaussian, the most common estimator, which is statistically optimal (BLUE), is the Least-Squares (LS) one. However, this is not the case in the presence of outliers or when the noise distribution exhibits long tails. In the following, we will give a brief overview of the algorithmic schemes that has been proposed to deal with the aforementioned problem. These schemes can be classified into two major categories, those that apply a weighted scheme to penalize the largest residuals and those that apply sparse modeling.

It should also be noted, that for the case where $n<m$ (underdetermined system of linear equations\footnote{An infinite number of solutions exist.}), an additional condition/constraint should be imposed, if one wishes to recover the unknown vector. One of the most common features lately, is to impose sparsity constraints on $\boldsymbol{\theta}_0$. However, in such a case, the task breaks down to a classical sparse model for the combined matrix $[ \mathbf{X}\ \boldsymbol{I} ]$, which has been extensively studied over the last years. Thus, any sparse related algorithm, e.g., ADMM for solving the LASSO formulation, OMP or any other greedy approach method, is rendered suitable for performing the estimation. Hence, the study of this case is considered as trivial.

\subsection{Penalizing Large Residuals}
\label{subsec:penlarres}
Both methods of this group, attempt to estimate $\boldsymbol\theta_0$, based on equation \eqref{eq:probdefmat}.
\begin{itemize}
\item \textbf{M-estimators (M-est) \cite{huber1981wiley}:}\\
  In M-est a robust cost function $\rho$ (satisfying certain properties) of the residual error
$r_i=y_i-\boldsymbol{x}_i^T\boldsymbol{\theta},\ i=1,2,...,n$,
is minimized, so that
$\boldsymbol{\theta}_*=\argmin_{\boldsymbol{\theta}} \sum_{i=1}^n \rho(r_i)$.
Differentiation with respect to $\boldsymbol\theta$ leads to $\sum_{i=1}^n \psi(r_i)\boldsymbol{x}_i^T= \mathbf{0}$, where $\psi=\rho'$.
If we define $w(r)=\psi(r)/r,\ w_i=w(r_i)$, then the system of normal equations is cast as $\sum_{i=1}^n w_i r_i\boldsymbol{x}_i^T= \mathbf{0}$. This is the basic version of M-est, although several variations exist. In our experiments, we have used a scaling parameter $\hat{\sigma}$, computed at each step, and defined $\rho:= \rho(r_i/\hat{\sigma})$. Consequently, another way to interpret M-est, is by solving a Weighted Least-Squares problem,
\begin{equation}
\min_{\boldsymbol{\theta}} \sum_{i=1}^n w_i r_i^2  \Leftrightarrow\min_{\boldsymbol{\theta}} ||\mathbf{W_{r}}^{1/2}(\boldsymbol{y}-\mathbf{X}\boldsymbol{\theta})||_2^2.
\label{eq:M-est}
\end{equation}
To solve (\ref{eq:M-est}) the {\em Iteratively reweighted least squares} (IRLS) algorithmic scheme is employed, where the diagonal weight matrix $\mathbf{W_{r}}$ assigns the weights, with values depending on the Robust function selected (Huber's, Tukey's biweight, Hampel's, Anrdews). Notice that if $\mathbf{W_{r}}=\mathbf{I}_n$, the scheme performs a classic Least-Squares. Many improved variants can be found in the literature. For more details, the interested reader is referred to \cite{huber1981wiley, rousseeuw2005robust, maronna2006robust}.

\item \textbf{Robust Orthogonal Matching Pursuit (ROMP) \cite{razavi2012robust}:}\\
The method is based on the M-est and the popular OMP algorithm, which was studied in \cite{davis1997adaptive, tropp2004greed, tropp2007signal}. However, the key aspect of the algorithm, which is also the feature that introduces robustness, is the execution of a weighted Least Squares step (M-est), instead of an ordinary one, each time the support set is augmented with an atom. Although a variety of termination criteria exist, we let the algorithm terminate, as soon as the length of the residual vector drops below a predefined threshold. The method is summarized as follows:
\textit{Initialization}: $k:=0$, $\Omega^{0}:= \emptyset $, $\boldsymbol\theta^{(0)}=\boldsymbol{0}$, $\boldsymbol{r}^{(0)}=\boldsymbol{y}$.\\ 
\textbf{Main iteration}\\
\textit{Step 1 - Initialization}:\\ $k:=k+1$, $\hat{\sigma}=MAD(\boldsymbol{r}^{(k-1)})$, $\boldsymbol{r}_{\psi}^{(k-1)}=\psi(\boldsymbol{r}^{(k-1)}/\hat{\sigma})$.\\
\textit{Step 2 - Atom selection}:\\ $i_k:= \argmax \left| \mathbf{X}^T\boldsymbol{r}_{\psi}^{(k-1)} \right| $, $\Omega^{k}:= \Omega^{k-1}\cup  i_k$.\\
\textit{Step 3 - Solution}:\\ 
\begin{equation}
\boldsymbol{\theta}^{(k)}:= \argmin_{\boldsymbol{\theta}} ||\mathbf{W_r}^{1/2}(\boldsymbol{y} - \mathbf{X}_{\Omega^k}\boldsymbol\theta)||_2^2,
\label{eq:ROMP}
\end{equation}
$\boldsymbol{r}^{(k)}= \boldsymbol{y} -\mathbf{X}_{\Omega^k}\boldsymbol\theta^{(k)}$.\\
The main procedure begins with the computation of $\hat{\sigma}=MAD(\boldsymbol{r})$\footnote{Median Absolute Deviation $MAD(\boldsymbol{x})=median_{i}(|x_i-median_{i}(x_i)|)$.} and the residual pseudo-values $\boldsymbol{r}_{\psi}$, which are then used for selecting the atom of matrix $\mathbf{X}$, that is most correlated to the residual pseudo-values (Step 2). Finally, it should be noted, that $\psi$ is a robust function (as in M-est), that also assigns the weights of matrix $\mathbf{W_r}$, required for solving the weighted Least Squares step, in \eqref{eq:ROMP}. Here, the difference to \eqref{eq:M-est}, is that at each $k$ step, $\mathbf{X}_{\Omega^k}$ includes only the columns of $\mathbf{X}$ that have been selected until the current step. Unfortunately, no theoretical justifications have been made, either on the selection of the atom based on the residual pseudo-values or on the iterative employment of the M-est.
\end{itemize}

\subsection{Sparse outlier modeling}
\label{subsec:sparsemod}
For all of the following methods, a different model is adopted. To this end, assume, that the outlier noise values are significantly fewer (i.e., sparse) compared to the size of the input data. Thus, a familiar technique, is to express the noise vector as a sum of two independent components, $\boldsymbol{e}=\boldsymbol{u}+\boldsymbol\eta$, where $\boldsymbol\eta$ is assumed to be the dense inlier noise vector of energy $\epsilon_0$ and $\boldsymbol{u}\in \Real^n$ the sparse outlier noise vector with support set, $T$, and cardinality $|T|\leq s<<n$. The support set is defined as the set of indices $i\in\{0,\dots,n\}$ that satisfy  $u_i\not=0$. Hence, equation \eqref{eq:probdefmat} can be recast as:
\begin{equation}
\boldsymbol{y} = \mathbf{X}\boldsymbol\theta_0+\boldsymbol{u}_0+ \boldsymbol\eta.
\label{eq:probdefmat2}
\end{equation}
As we would like to minimize the number of outliers in \eqref{eq:probdefmat2}, the associated optimization problem becomes:
\begin{equation}
\min_{\boldsymbol\theta,\boldsymbol{u}} ||\boldsymbol{u}||_0, \ \text{s.t.}\  ||\boldsymbol{y} - \mathbf{X}\boldsymbol\theta-\boldsymbol{u}||_2 \leq \epsilon_0.
\label{eq:probformL0}
\end{equation}

However, in general, the task in (\ref{eq:probformL0}) is  a combinatorial problem. Hence, many authors propose to relax the $\ell_0$ with the $\ell_1$ norm, using a similar formulation:
\begin{equation}
\min_{\boldsymbol\theta,\boldsymbol{u}} ||\boldsymbol{u}||_1, \ \text{s.t.}\  ||\boldsymbol{y} - \mathbf{X}\boldsymbol\theta-\boldsymbol{u}||_2 \leq \epsilon_0,
\label{eq:probformL1a}
\end{equation}
This has the advantage of transforming (\ref{eq:probformL0}) to a convex problem, which can be solved using a variety of methods.

\begin{itemize}
\item \textbf{LASSO formulation for robust denoising \cite{boyd2011alternating,boyd2011distributed,mateos2012robust}:}\\
The Alternating Direction Method of Multipliers (ADMM) is a technique for solving the Lagrangian form of (\ref{eq:probformL1a}), for appropriate  multiplier values $\lambda>0$ (Generalized Lasso form):
\begin{equation}
\boldsymbol{w}_*:= \argmin_{\boldsymbol{w}} \{ (1/2)||\boldsymbol{y} - \mathbf{A}\boldsymbol{w}||_2^2 + \lambda ||\mathbf{F}\boldsymbol{w} ||_1 \},
\label{eq:genlasso}
\end{equation}
where $\boldsymbol{w} =(\boldsymbol\theta, \; \boldsymbol{u})^T$, $\mathbf{A}=[\mathbf{X}\ \mathbf{I}_n ],$ $\mathbf{F}= [\mathbf{O}_{n\times m}\  \mathbf{I}_n]$ (for $\mathbf{F}=\mathbf{I}_{n+m},$ we have the standard Lasso form). The ADMM method was studied in the 70's and 80's, as a good alternative to penalty methods, although it was established as a method to solve partial differential equations, \cite{peaceman1955numerical, douglas1955numerical}. 


\item \textbf{Second Order Cone Programming (SOCP):}\\
Problem \eqref{eq:probformL1a} is also known as \textit{Robust Regression Basis Pursuit}-(BPRR) and it can be reformulated as a Second Order Cone Programming (SOCP) task \cite{lobo1998applications, boyd2004convex}:
\begin{align}
\begin{matrix}\label{eq:SOCP}
  & \boldsymbol{w}_* := \argmin_{\boldsymbol{w}} \boldsymbol{g}^T\boldsymbol{w}, \\
\text{s.t.}\ & \mathbf{H}^T\boldsymbol{w} \geq \mathbf{0},\ \boldsymbol{y}- \mathbf{R}\boldsymbol{w} \in \mathcal{C}_{\epsilon_0}^{n+1}
\end{matrix}
\end{align}
where $\boldsymbol{g}=(\boldsymbol{0},\; \boldsymbol{0},\; \boldsymbol{1})^T\in \Real^{m+2n}$,
$\boldsymbol{w} = (\boldsymbol\theta,\; \boldsymbol{u},\; \boldsymbol{s})^T$,
\begin{align*}
\mathbf{H}=\begin{bmatrix}
\mathbf{O}_{m\times n} & \mathbf{O}_{m\times n} \\
-\mathbf{I}_n & \mathbf{I}_n\\
\mathbf{I}_n & \mathbf{I}_n\\
\end{bmatrix},\;
\mathbf{R}= [\mathbf{X}\  \mathbf{I}_n\ \mathbf{O}_{n\times n}]
\end{align*}
and $\mathcal{C}_{\epsilon_0}^{n+1}$ is the unit second order (convex) cone of dimension $n+1$.

\item \textbf{Sparse Bayesian Learning (SBL) \cite{tipping2001sparse, wipf2004sparse, jin2010algorithms,theodoridis2015machine_learning}:}\\
Another path that has been exploited in the respective literature is to use Sparse Bayesian Learning techniques \cite{mitra2010robust, jin2010algorithms}. Assume that $u_i$ is a random variable with prior distribution $u_i\sim \mathcal{N}(0, \gamma_i)$,
where $\gamma_i$ is the hyperparameter that controls the variance of each $u_i$ that has to be learnt. If $\gamma_i = 0$, then $u_i=0$, i.e., no outlier exists on this index. In contrast, a positive value of $\gamma_i$, results in an outlier in the measurement $i$.
To estimate the regression coefficients, we jointly find
\begin{equation}
(\boldsymbol\theta_*, \boldsymbol\gamma_*, \sigma_*^2) =  \argmax_{\boldsymbol\theta,\boldsymbol\gamma,\sigma^2} P(\boldsymbol{y}| \mathbf{X},\boldsymbol\theta,\boldsymbol\gamma,\sigma^2 ),
\label{eq:Bayes1}
\end{equation}
where $\boldsymbol\gamma := (\gamma_1, \gamma_2,...,\gamma_n)^T$ and $\eta_i \thicksim  \mathcal{N}(0,\sigma^2)$. The posterior estimation of $\boldsymbol{u}$, follows, from:
\begin{equation}
\boldsymbol{u}_* = E[\boldsymbol{u}| \mathbf{X},\boldsymbol\theta_*, \boldsymbol\gamma_*, \sigma_*^2 ].
\label{eq:Bayes2}
\end{equation}
%
\end{itemize}


\section{Greedy Algorithm for Robust Denoising (GARD)}
\label{sec:GARD}

The goal of the proposed algorithmic scheme is to solve problem \eqref{eq:probformL0} using the split noise model described in \eqref{eq:probdefmat2} and it is  designed along the celebrated Orthogonal Matching Pursuit rationale. It should be noted that ROMP, which also employs OMP's selection technique, is quite different from our approach. ROMP is mainly based on the M-est algorithm, while the proposed scheme, in contrast to other methods, tackles directly problem \eqref{eq:probformL0} and alternates between a least squares minimization task and an OMP selection technique. It can be easily seen that \eqref{eq:probformL0} can also be cast as:
\begin{equation}
\min_{\boldsymbol\theta,\boldsymbol{u}} ||\boldsymbol{u}||_0, \ \text{s.t.}\  \Big\| \boldsymbol{y} - \mathbf{A}\begin{pmatrix}
\boldsymbol\theta\\
\boldsymbol{u}
\end{pmatrix}\Big\|_2 \leq \epsilon_0,
\label{eq:probformL0tog}
\end{equation}
where $\mathbf{A}=[\mathbf{X}\ \mathbf{I}_n]$.
Following OMP's rationale, at each iteration step, GARD estimates the solution, i.e., $\boldsymbol{z}_{*}^{(k)}=(\boldsymbol{\theta}_{*}^{(k)} , \boldsymbol{u}_{*}^{(k)})^T \in \Real^{m+k}$ (for step $k=0$ no outlier estimates exist and $\boldsymbol{z}_{*}^{(0)} \in \Real^m $), using a Least-Squares criterion (i.e., $\min_{\boldsymbol{\theta}, \boldsymbol{u}}\|\boldsymbol{y} - \mathbf{A}(\boldsymbol{\theta} , \boldsymbol{u})^T\|^2$). In the following iterations GARD selects the observation which is the furthest away from the solution, using OMP's selection rationale (based on correlation).
Hence, GARD restricts the selection over atoms of the second half of matrix $\mathbf{A}$, i.e., matrix $\mathbf{I}_n=[\boldsymbol{e}_1\ \boldsymbol{e}_2\ ...\ \boldsymbol{e}_n]$, where $\boldsymbol{e}_i$ are the vectors of the standard basis of $\Real^n$.

To be more specific, at the first step the algorithm computes the initial Least-Squares solution disregarding the presence of outliers, i.e., $\boldsymbol{\theta}_* = \argmin_{\boldsymbol{\theta}}\|\boldsymbol{y} - \mathbf{X}\boldsymbol{\theta}\|^2$, and the initial residual $\boldsymbol{r}^{(0)} = \boldsymbol{y} - \mathbf{X}\boldsymbol{\theta}_*$. Moreover, the set of the so called {\em active columns} is initialized to include all the columns of $\mathbf{X}$. Then, the main iteration cycle begins. At each step, GARD can be divided into two parts:
\begin{itemize}
\item Firstly, the greedy selection step is performed, i.e., the column vector from matrix $\mathbf{I}_n$ that is more correlated with the latest residual is selected and the set of active columns of $\mathbf{A}$ (i.e., the set of columns that have already been selected) is augmented by that column. The correlation is measured with respect to the angle, which in turn leads to the maximization of $\left| \langle \boldsymbol{r}^{(k)},\boldsymbol{e}_i \rangle \right|= \left| r_{i}^{(k)} \right|$ for an index $i \in J=1\textrm{---}n$.
\item Next, a Least-Squares solution step is performed and the new residual is computed.
\end{itemize}
This procedure is repeated until the residual drops below a specific predefined threshold, as described in details in Algorithm \ref{algo:GARD}. In order to avoid any confusion, we should also emphasize that even though both the sets $J$ and $S_{inac}$ correspond to the same orthonormal vectors $\boldsymbol{e}_i$ of matrix $\mathbf{I}_n$, they should not be regarded as equal, since 
$J$ includes indices from $\mathbf{I}_n$, whereas $S_{inac}$ includes indices from the second half of the augmented matrix $\mathbf{A}$. Consequently, since $\boldsymbol{r}^{(k-1)} \in \Real^n$ and the index selected, i.e., $j_k \in S_{inac}$, exceeds the dimensions of $\mathbf{I}_n$, the index $j-m$ is used for the elements of $\boldsymbol{r}^{(k-1)}$ in order to include indices that belong to the set $J$. For instance, if the largest component of $|r^{(k-1)}|$ is the $5$-th one, this leads to the fact that $j_k -m=5 \in J$ and $j_k=m+5 \in S_{inac}$.
As it will be shown, the improved performance of the proposed scheme is due to the orthogonality between the columns of $\mathbf{I}_n$ (standard Euclidean basis).
\begin{algorithm}
\caption{: Greedy Algorithm for Robust Denoising (GARD)} \label{algo:GARD}
\begin{algorithmic}
\STATE{\textbf{Input}}: $\mathbf{X},\ \boldsymbol{y},\ \epsilon_0$
\STATE{\textbf{Output}}: $\boldsymbol{z}_*=(\boldsymbol\theta_*,\boldsymbol{u}_*)^T$
 \STATE{Initialization}: $k:=0$\\
	$S_{ac}=\{1,2,...,m\}$, $S_{inac}=\{m+1,...,m+n\}$\\
	$\mathbf{A}_{ac}^{(0)}=\mathbf{X}$\\
  \STATE{\textbf{Solution*}}:
	$\boldsymbol{z}_*^{(0)} := \argmin_{\boldsymbol{z}} ||\boldsymbol{y} - \mathbf{A}_{ac}^{(0)}\boldsymbol{z}||_{2}^2$ \\
	\STATE{Initial Residual}:  $\boldsymbol{r}^{(0)}=\boldsymbol{y} - \mathbf{A}_{ac}^{(0)}\boldsymbol{z}_*^{(0)}$\\
 \WHILE {$||\boldsymbol{r}^{(k)} ||_2 > \epsilon_0$}
 \STATE $k:=k+1$\\
\STATE{Selection}: $j_k:=\argmax_{ j \in S_{inac}} |r_{j-m}^{(k-1)}|$\\
\STATE{Update Support}: $S_{ac}:=S_{ac}\cup \{j_k\},\ S_{inac}=S_{ac}^c,\ \mathbf{A}_{ac}^{(k)}= [\mathbf{A}_{ac}^{(k-1)}\ \boldsymbol{e}_{j_k}]$\\
\STATE{\textbf{Update Solution**}}: $\boldsymbol{z}_*^{(k)}:= \argmin_{\boldsymbol{z}} ||\boldsymbol{y} - \mathbf{A}_{ac}^{(k)}\boldsymbol{z}||_{2}^2$\\
\STATE{Update Residual}: $\boldsymbol{r}^{(k)}=\boldsymbol{y} - \mathbf{A}_{ac}^{(k)}\boldsymbol{z}_*^{(k)}$
\ENDWHILE
\end{algorithmic}
\end{algorithm}
The complexity of the algorithm is $O \big( (m+k)^3/3+n(m+k)^2 \big)$ at each $k$ step, making it unattractive when the dimension of the unknown vector is large. However, since at each step the method solves a standard Least-Squares problem, the complexity could be further reduced using \textit{Cholesky} decomposition, \textit{QR} factorization or the \textit{matrix inversion Lemma}. For details on those implementations of the classic OMP, read \cite{sturm2012comparison}. Playing with all schemes, we found that in our case the most efficient implementation was the Cholesky decomposition, as described below:
\begin{itemize}
\item \textit{Replace} the initial solution step $k:=0$, of Algorithm \ref{algo:GARD}, with:
\begin{algorithmic}
\STATE{\textbf{Solution*}}:\\
	\STATE{Factorization step}: $\mathbf{W}_{0}=\mathbf{X}^{T} \mathbf{X}$ \text{as}
	$\mathbf{W}_{0}=\mathbf{L}_{0} \mathbf{L}_{0}^{T}$.\\
  \STATE{Solve}
	$\mathbf{L}_{0}\mathbf{L}_{0}^{T}\boldsymbol{z}=\mathbf{X}^{T}\boldsymbol{y}$ \text{using}:\\
	\begin{itemize}
	\item \text{forward substitution} $\mathbf{L}_{0}\boldsymbol{q}=\mathbf{X}^{T}\boldsymbol{y}$
	\item \text{backward substitution} $\mathbf{L}_{0}^{T}\boldsymbol{z}_*^{(0)}=\boldsymbol{q}$.
	\end{itemize}
 \end{algorithmic}

\item \textit{Replace} the update solution step $k:=k+1$, of Algorithm \ref{algo:GARD}, with:
\begin{algorithmic}
\STATE{\textbf{Update Solution**}}:\\
\STATE{Compute} $\boldsymbol{v}$ \text{such that}: $\mathbf{L}_{k-1}\boldsymbol{v}=\mathbf{A}_{ac}^{(k-1)^T}\boldsymbol{e}_{j_k}$
\STATE{Compute}: $b=\sqrt{1-||\boldsymbol{v}||_2^2}$
\STATE{Matrix Update}: $ \mathbf{L}_{k}=\begin{pmatrix}
\mathbf{L}_{k-1} & \mathbf{0}\\ \boldsymbol{v}^T & b
\end{pmatrix} $
 \STATE{Solve}
	$\mathbf{L}_{k}\mathbf{L}_{k}^{T}\boldsymbol{z}=\mathbf{A}_{ac}^{(k)^T}\boldsymbol{y}$ \text{using}:\\
	\begin{itemize}
	\item \text{forward substitution} $\mathbf{L}_{k}\boldsymbol{p}=\mathbf{A}_{ac}^{(k)^T}\boldsymbol{y}$
	\item \text{backward substitution} $\mathbf{L}_{k}^{T}\boldsymbol{z}_*^{(k)}=\boldsymbol{p}$.
	\end{itemize}
\end{algorithmic}
\end{itemize}

This modification leads to a reduction of the cost by an order of magnitude, for the main iteration steps. Analytically, the cost at the initial factorization, plus the cost for the forward and backward substitution, is $O(m^3/3+nm^2)$. At each next step, neither inversion nor factorization is required. The lower triangular matrix $\mathbf{L}_{k}$ is updated, only with a minimal cost of square-dependence. Furthermore, the cost required for solving the linear system using forward and backward substitution at each next $k$ step is $O((m+k)^2+2n(m+k))$.  Thus, the \textit{total complexity} of the efficient  (via the Cholesky decomposition) GARD implementation is $O\big(m^3/3 + k^3/2 +(n+3k)m^2 + 3kmn \big)$.

\begin{remark}
The algorithm begins with a Least-Squares solution to obtain $\boldsymbol{z}_*^{(0)}$. Thus, if no outliers exist, GARD solves the standard Least-Squares problem; it provides the maximum likelihood estimator (MLE), when the noise is Gaussian.
\end{remark}
\begin{remark}
Since GARD solves a Least-Squares problem at each step, the new residual, $\mathbf{r}^{(k)}$, is orthogonal to each column that participates in the representation, i.e., $\langle \mathbf{r}^{(k)}, \boldsymbol{e}_{j_k} \rangle = r_{j_k}^{(k)} = 0,\ \forall\  k=1,2,\dotsc $. Thus, column  $\boldsymbol{e}_{j_k}$ of matrix $\mathbf{I}_n$, cannot be reselected.
\end{remark}
\begin{remark}
Considering the complexity of the efficient implementation of GARD, the algorithm speeds up in cases where the fraction of the outliers is very low, i.e., the outlier vector is very sparse $(k<<n)$.
\end{remark}
\begin{remark}
Matrix $\mathbf{A}_{ac}^{(k)}=[\mathbf{A}_{ac}^{(k-1)} \boldsymbol{e}_{j_k}]$ could also be cast as
$\mathbf{A}_{ac}^{(k)}= \begin{bmatrix}
 \mathbf{X}\  \mathbf{I}_{S_k}
\end{bmatrix}$,
where $S_k=\{j_1,j_2,...,j_k \}$, is the set of columns selected at the current step, i.e., the support set of our sparse estimate. Thus, the estimated support should not be confused with the set of active columns $S_{ac}$ that participate in the representation of $\boldsymbol{y}$.
\end{remark}
\begin{remark}The proposed scheme should not be confused with other OMP-based schemes, such as Robust OMP in \cite{razavi2012robust}; although both are OMP-based, they perform in a distinctive manner and for dissimilar purposes. As both the selection step, as well as the minimization step work quite different, GARD selects a column, based on the residual and performs a classic Least Squares procedure, whereas ROMP selects a column based on the residual pseudo-values and then solves a weighted Least Squares minimization problem.
\end{remark}

\section{Theoretical results}
\label{sec:theoretic}
This section is devoted to study the main properties of GARD. Firstly, the convergence properties of the proposed scheme are derived. In the sequel, it is shown that GARD can recover the exact solution, under certain assumptions, in the presence of outlier noise only. Finally, for the case of both inlier and outlier noise, bounds for the recovery of the sparse outlier support, as well as the reconstruction error are presented.

\subsection{General results}
\label{subsec:general_res}
\begin{lemma}
At every $k\leq n-m$ step, GARD selects a column vector $\mathbf{e}_{j_k}$ from matrix $\mathbf{I}_n$, that is linearly independent of all the column vectors in matrix $\mathbf{A}_{ac}^{(k-1)}$. Hence, $\mathbf{A}_{ac}^{(k)}$ has full rank and the solution to the Least-Squares problem at each step is unique.
\label{lem:fullrank}
\end{lemma}
%

\begin{proof}
The proof relies on mathematical induction. At the initial step, matrix $\mathbf{A}_{ac}^{(0)}=\mathbf{X}$ has been assumed to be full rank, hence the solution of the Least-Squares problem is unique. Suppose that at $k-1$ step ($k \in \mathbf{N}^{*}$), matrix $\mathbf{A}_{ac}^{(k-1)}$ is full rank, hence let $\boldsymbol{z}_*^{(k-1)}$ denote the unique solution of the Least-Squares problem and $\boldsymbol{r}^{(k-1)}=\boldsymbol{y}-\mathbf{A}_{ac}^{(k-1)}\boldsymbol{z}_*^{(k-1)}$ the residual at the current step. Assume that at the $k-$th step, the $j_k-$th column of matrix $\mathbf{I}_n$ is selected from the set $S_{inac}$. We will prove that the columns of the augmented matrix at this step, i.e., the columns of matrix $\mathbf{A}_{ac}^{(k)}=[\mathbf{A}_{ac}^{(k-1)}\ \boldsymbol{e}_{j_k}]$, are linearly independent.  Since $j_k:=\argmax_{j \in S_{inac}} |r_{j}^{(k-1)}|$, we have that $r_{j_k}^{(k-1)}\neq 0$ (otherwise either this wouldn't have been selected or the residual vector would be equal to zero). Suppose that, the columns of matrix $\mathbf{A}_{ac}^{(k)}$ are linearly dependent, i.e., there exists $\boldsymbol{a}\neq \mathbf{0}$, such that $\boldsymbol{e}_{j_k}=\mathbf{A}_{ac}^{(k-1)}\boldsymbol{a}$ and let $\tilde{\boldsymbol{z}}^{(k-1)}=\boldsymbol{z}_*^{(k-1)}+r_{j_k}^{(k-1)}\boldsymbol{a}$.
Thus, we have
\begin{align*}
||\tilde{\boldsymbol{r}}^{(k-1)} ||_2 &= || \boldsymbol{y}-\mathbf{A}_{ac}^{(k-1)}\tilde{\boldsymbol{z}}^{(k-1)}||_2 = \\
&= || \boldsymbol{y}-\mathbf{A}_{ac}^{(k-1)}\boldsymbol{z}_*^{(k-1)} - r_{j_k}^{(k-1)} \mathbf{A}_{ac}^{(k-1)} \boldsymbol{a}  ||_2 = \\
&= ||\boldsymbol{r}^{(k-1)} - r_{j_k}^{(k-1)}\boldsymbol{e}_{j_k} ||_2
< ||\boldsymbol{r}^{(k-1)} ||_2,
\label{fullrank}
\end{align*}
which contradicts the fact that the residual of the Least-Squares solution attains the smallest norm. Thus, all the selected columns of matrix $\mathbf{A}_{ac}^{(k)}$  are linearly independent.
\end{proof}
\begin{theorem}
The norm of the residual vector \[\boldsymbol{r}^{(k)} = \boldsymbol{y} - \mathbf{A}_{ac}^{(k)}\boldsymbol{z}_*^{(k)}\] in GARD is strictly decreasing. Moreover, the algorithm
will always converge.
\label{theor:errordec}
\end{theorem}
\begin{proof}
Let $\boldsymbol{z}_*^{(k-1)}$ be the unique Least-Squares solution (Lemma \ref{lem:fullrank}) and $\boldsymbol{r}^{(k-1)} = \boldsymbol{y} - \mathbf{A}_{ac}^{(k-1)}\boldsymbol{z}_*^{(k-1)}$ the respective residual at $k-1$ step. At the next step, the algorithm selects the column $j_k$ and augments matrix $\mathbf{A}_{ac}^{(k-1)}$, by column $\boldsymbol{e}_{j_k}$ to form matrix $\mathbf{A}_{ac}^{(k)}$. Let $\boldsymbol{z}_*^{(k)}$ denote the unique solution of the Least-Squares problem at the $k-$th step (Lemma \ref{lem:fullrank}) and $\boldsymbol{r}^{(k)} = \boldsymbol{y} - \mathbf{A}_{ac}^{(k)}\boldsymbol{z}_*^{(k)}$ the respective residual. Consequently, one could define a cost function  for every $\boldsymbol{z} \in \Real^{m+k}$ at $k$ step, as
$P^{(k)}(\boldsymbol{z})=|| \boldsymbol{y} - \mathbf{A}_{ac}^{(k)}\boldsymbol{z}||_2. $
 Thus, we have that
\begin{equation}
||\boldsymbol{r}^{(k)}||_2 = P^{(k)}(\boldsymbol{z}_*^{(k)}) \leq P^{(k)}(\boldsymbol{z}),
\label{eq:proof1}
\end{equation}
for every $\boldsymbol{z} \in \Real^{m+k}$. Now let $\mathbf{\mathtt{z}}^{(k)}=(\boldsymbol{z}_*^{(k-1)},r_{j_k}^{(k-1)})^T$, where $r_{j_k}^{(k-1)}$ is the $j_k$ coordinate of the residual $\boldsymbol{r}^{(k-1)}$. Thus, we have that
\begin{align}
P^{(k)}(\mathbf{\mathtt{z}}^{(k)}) &= || \boldsymbol{y} - \mathbf{A}_{ac}^{(k)}\mathbf{\mathtt{z}}^{(k)} ||_2 \nonumber \\ &= || \boldsymbol{y} - \mathbf{A}_{ac}^{(k-1)}\boldsymbol{z}_*^{(k-1)} - r_{j_k}^{(k-1)}\boldsymbol{e}_{j_k}||_2 \nonumber \\
&= ||\boldsymbol{r}^{(k-1)} - r_{j_k}^{(k-1)}\boldsymbol{e}_{j_k} ||_2 < ||\boldsymbol{r}^{(k-1)}  ||_2.
\label{eq:proof2a}
\end{align}
Combining \eqref{eq:proof1} and \eqref{eq:proof2a}, we have that
\begin{equation}
||\mathbf{r}^{(k)}  ||_2 < ||\mathbf{r}^{(k-1)}  ||_2.
\label{eq:proof2b}
\end{equation}
Since $\boldsymbol{y} \in \Real^n$, the residual equals zero, as soon as $n-m$ columns have been selected. However, since the noise bound is a positive value assumed to be known, the algorithm terminates at the first step $k<n-m$, where the residual's norm drops below $\epsilon_0$.
\end{proof}

\subsection{The presence of outliers only}
\label{subsec:onlyoutliers}
The scenario where the signal is corrupted only by outliers is treated separately. In this case, we aim to solve the following $\ell_0$ minimization problem:
\begin{align}
\begin{matrix}
\min_{\boldsymbol\theta, \boldsymbol{u}}  &  ||\boldsymbol{u}||_{0}\\
\text{s. t.} &  \boldsymbol{y} = \mathbf{X}\boldsymbol\theta+\boldsymbol{u},
\end{matrix}
\label{eq:probdefmat_only_out}
\end{align}
where $\mathbf{X}$ is assumed to be a full column rank matrix (note that if $\mathbf{X}$ has linearly dependent columns, there is not a unique solution for this problem). The general solution  of \eqref{eq:probdefmat_only_out} is an NP-hard task. However, under specific conditions, the problem can be solved efficiently using GARD, as it will be proved subsequently.

To simplify notation and reduce the size of the subsequent proofs, we  orthonormalize $\mathbf{X}$ by the reduced $\textit{QR}$ decomposition, i.e., $\mathbf{X}=\mathbf{QR}$, where $\mathbf{Q}$ is a $n \times m $ matrix, whose columns form an orthonormal basis of the column space of $\mathbf{X}$ (i.e., $\SPAN(\mathbf{X})$) and $\mathbf{R}$ is a $m \times m$ upper triangular matrix. Since $\mathbf{X}$ has full column rank, the decomposition is unique; moreover, matrix $\mathbf{R}$ is invertible. Using this decomposition, the split noise modeling described in equation \eqref{eq:probdefmat2} can be written as
\begin{equation}
\boldsymbol{y} = \mathbf{Q}\boldsymbol{w}_0+\boldsymbol{u}_0 +\boldsymbol\eta,
\label{eq:probdefmat_qr_noise}
\end{equation}
where $\boldsymbol{w}_0=\mathbf{R}\boldsymbol\theta_0.$ If $\boldsymbol{w}_0$ is recovered, the unknown vector $\boldsymbol\theta_0$ can also be recovered from $\boldsymbol\theta_0=\mathbf{R}^{-1}\boldsymbol{w}_0$. Equation \eqref{eq:probdefmat_qr_noise} plays a central role, as it describes the model that is adopted throughout this paper. In this section, however, we assume that only outlier noise exist, hence $\boldsymbol{\eta}$ is set to zero.

We are now in the position to express equation \eqref{eq:probdefmat_qr_noise} (for $\boldsymbol{\eta}=\boldsymbol{0}$) as 
$\boldsymbol{y} =  [\mathbf{Q}\ \mathbf{I}_n] \boldsymbol{z}'_0,$
where $\boldsymbol{z}'_0 =(\boldsymbol{w}_0^T ,\boldsymbol{u}_0^T)^{T}.$ Since the vector $\boldsymbol{u}_0$ is $s$-sparse at most, the measurement vector could also be written as $
\boldsymbol{y} = [\mathbf{Q}\ \mathbf{I}_{S}] \boldsymbol{z}_0,$ where $\boldsymbol{I}_S$ is the matrix containing column vectors from $\boldsymbol{I}_n$ indexed by\footnote{Recall that $\boldsymbol{I}_S=[\boldsymbol{e}_{j_1}, \dots \boldsymbol{e}_{j_{|S|}}]$, where $i_1<\dots<i_{|S|}$ are the indices of $S$.} $S$ and $\boldsymbol{z}_0=(\boldsymbol{w}_0^T, [\boldsymbol{u}_0]_S^T)^{T}$, with $[\boldsymbol{u}_0]_S \in \Real^{s}$ representing the reduced vector\footnote{Recall that $\boldsymbol{u}=\boldsymbol{I}_S\boldsymbol{u}_S$.}, that contains only the non-zero entries of $\boldsymbol{u}_0$. 

We assume that the outlier vector is sparse over the support subset $S\subset 1\textrm{---}n$, with $|S|= s<<n$ (i.e., $u_i=0$, for all $i\not\in S$ and $u_i\not=0$ for $i\in S$) and that $s<n/2$ (in the case where $s\geq n/2$, the solution cannot be recovered \cite{candes2005decoding}). Applying the \textit{QR} decomposition, problem \eqref{eq:probdefmat_only_out} could also be written as:
\begin{align}
\begin{matrix}
\min_{\boldsymbol{w}, \boldsymbol{u}}  &  ||\boldsymbol{u}||_{0}\\
\text{s. t.} &  \boldsymbol{y} = \mathbf{Q}\boldsymbol{w}+\boldsymbol{u},
\end{matrix}
\label{eq:probdefmat_only_out2}
\end{align}

In the following, the notion of the \textit{smallest principal angle} between subspaces is employed.  Given the information concerning the index subset $S$ (i.e., we assume that we know the support of the outliers),  $\boldsymbol{w}$ can be recovered, if and only if $ [\mathbf{Q}\ \mathbf{I}_S]$ has full rank. The latter assumption can also be expressed in terms of the \textit{smallest principal angle}, $\omega_S$, between the subspace spanned by the columns of the regressor matrix, i.e.,  $\SPAN(\mathbf{Q})$ and the subspace spanned by the columns of $\boldsymbol{I}_S$, i.e., $\SPAN(\mathbf{I}_S)$.

\begin{definition}\label{DEF:pa}
Let $\delta_{S}$ be the smallest number that satisfies the inequality $| \langle \boldsymbol{w},\boldsymbol{u} \rangle | \leq \delta_S ||\boldsymbol{w}||_2 ||\boldsymbol{u}||_2$, for all $\boldsymbol{w} \in \SPAN(\mathbf{Q})$ and $\boldsymbol{u} \in \SPAN(\mathbf{I}_S)$. Then $\omega_S=\arccos(\delta_S)$ is the smallest principle angle between the spaces $\SPAN(\mathbf{Q})$ and $\SPAN(\mathbf{I}_S)$.
\end{definition}
Generalizing Definition \ref{DEF:pa}, we can take
\begin{align}
\delta_s=\max\left\{\delta_S,\; \textrm{for all } S\in\left(\begin{matrix}1\textrm{---}n\\ k\end{matrix}\right), \; k\leq s\right\},\label{EQ:principal_angle_final}
\end{align}
where $1\textrm{---}n=\{1,\dots,n\}$ and $\left(\begin{matrix}1\textrm{---}n\\ k\end{matrix}\right)$ denotes the set of all possible $k$-combinations of $1\textrm{---}n$. Hence, we define the smallest principal angle between the regression subspace $\SPAN(\mathbf{Q})$ and all the at most $s$-dimensional outlier subspaces; i.e., the spaces $\SPAN(\mathbf{I}_S)$ for all possible combinations of $S$, such that $|S|\leq s$, as follows:
\begin{equation}
\omega_s = \arccos(\delta_s).
\label{eq:principalangle2}
\end{equation}
It can readily be seen that $\delta_s$ can be defined by employing only the value $k=s$ (instead of all $k\leq s$) and that
for any  $\boldsymbol{w} \in \SPAN(\mathbf{Q})$ and any at most $s-$sparse vector $\boldsymbol{u}$ we have that
\begin{equation}
| \langle \boldsymbol{w},\boldsymbol{u} \rangle | \leq \delta_{s} ||\boldsymbol{w}||_2 ||\boldsymbol{u}||_2.
\label{eq:principalangle3}
\end{equation}
\begin{remark}
The quantity $\delta_s \in [0,1]$ (or equivalently $\omega_s \in [0^{\circ},90^{\circ}]$) is  a measure of how well separated the regressor subspace is from all the $s$-dimensional outlier subspaces.
\end{remark}

\noindent The following condition, also well known as the Restricted Isometry Property (R.I.P.) and plays a central role in sparse optimization methods.
\begin{definition}
For orthonormal matrix $\mathbf{Q}$, we define a constant $\mu_s$, $s=1,2,...,N$, as the smallest number such that
\begin{equation}
(1-\mu_s)||\boldsymbol{\alpha}||_2^2 \leq ||[\mathbf{Q}\ \mathbf{I}_S]\boldsymbol{\alpha}||_2^2 \leq (1+\mu_s) || \boldsymbol{\alpha}||_2^2.
\label{eq:rip}
\end{equation}
\label{def:rip}
\end{definition}
\noindent In \cite{mitra2013analysis} (Lemma III.1), it has been proved that for orthonormal regressor matrix $\mathbf{Q}$, the smallest principal angle coincides with the R.I.P. constant defined, i.e., $\delta_s=\mu_s,\ s=1,2,...,n$. Finally, the following theorem (\cite{candes2005decoding,mitra2013analysis}), guarantees uniqueness of the decomposition.

\begin{theorem}\label{THE:unique}
Assume that the vector $\mathbf{y}\in\Real^n$ can be decomposed as follows:
\begin{equation}
\mathbf{y} = \mathbf{Q}\boldsymbol{w}_0 + \boldsymbol{u}_0,
\label{eq:uniq_dec}
\end{equation}
where $\boldsymbol{w}_0\in\Real^m$ and $\boldsymbol{u}_0$ is an at most $s-$sparse vector. If $\delta_{2s}<1$ then this decomposition is unique.
\end{theorem}

One of the main theoretical results, established in this work is the following theorem, which guarantees the recovery of the support of the sparse vector, which in turn leads to the recovery of the exact solution for the case only outliers exist.
\begin{theorem}
Let $\mathbf{X}$ be a full column rank matrix and assume that the measurement vector, $\boldsymbol{y} = \mathbf{X}\boldsymbol{\theta}_0+\boldsymbol{u}_0$, has a unique decomposition, such that $||\boldsymbol{u}_0 ||_0 \leq s$ (at most $s$ outliers exist in the $\boldsymbol{y}$ variable). If
\begin{equation}
\delta_{s}<\sqrt{\frac{\min \{|u_i|,\; u_i\not=0\}}{2||\boldsymbol{u}_0||_2}},
\label{eq:delta_bound_definition}
\end{equation}
where $u_i$ are the elements of $\boldsymbol{u}_0$, then GARD guarantees that the unknown vector $\boldsymbol{\theta}_0$ and the sparse outlier vector $\boldsymbol{u}_0$ are recovered without any error.
\label{theor:exact_rec}
\end{theorem}
The proof of the theorem is found in Appendix \ref{appendix:A}.

%

\begin{remark}
The condition under which the measurement vector $\boldsymbol{y}$ can be uniquely decomposed into parts $\boldsymbol{y}_0=\mathbf{X}\boldsymbol{\theta}_0=\mathbf{Q}\boldsymbol{w}_0$ plus $\boldsymbol{u}_0$, is given in Theorem \ref{THE:unique} (see also \cite{candes2005decoding,mitra2013analysis}).
\end{remark}
\begin{remark}
The bound found in \eqref{eq:delta_bound_definition}, has also a nice geometrical interpretation. The ratio $\min \{ |u_i|,\; u_i\not=0\}/||\boldsymbol{u}_0||_2$, corresponds to the cosine of the largest direction angle of vector $\boldsymbol{u}_0$. Moreover, it can readily be seen that this ratio is no greater than $1$, which means that the right hand side of \eqref{eq:delta_bound_definition} is bounded by $\sqrt{2}/2$. In other words, the condition of Theorem \ref{theor:exact_rec} forces $\omega_s$ to lie in the interval $(45^{\circ} ,90 ^{\circ}]$.
\end{remark}

\subsection{The presence of both inlier and outlier noise}
\label{subsec:inl_outl}
The following results show that when GARD recovers the support of the outlier vector, the approximation error is relatively small.
\begin{theorem}
Let $\mathbf{X}$ be a full column rank matrix and assume that $\boldsymbol{y} = \mathbf{X}\boldsymbol{\theta}_0+\boldsymbol{u}_0+ \boldsymbol{\eta}$, such that $||\boldsymbol{u}_0 ||_0 \leq s$ (at most $s$ outliers exist in the $\boldsymbol{y}$ variable) and also $\left\| \boldsymbol{\eta} \right\|_2 \leq \epsilon_0 $. If
\begin{equation}
\delta_{s}< \sqrt{\frac{\min \{|u_i|,\; u_i\not=0\}-(2+\sqrt{6})\epsilon_0}{2  ||\boldsymbol{u}_0||_2 }},
\label{eq:delta_bound_definition_epsilon}
\end{equation}
where $u_i$ are the elements of $\boldsymbol{u}_0$, $d=\lceil \frac{n}{s}\rceil$, then GARD guarantees that the support of the sparse outlier vector $\boldsymbol{u}_0$ is recovered \footnote{Recall the definition of ceiling, i.e., $\lceil x \rceil = \min\{n \in \Integ\ | \ n\geq x  \} $.}.
\label{theor:exact_support}
\end{theorem}
The proof of the theorem is found in Appendix \ref{appendix:B}.
\begin{lemma}
\label{lem:singular_bound}
Assume that there exist $0 \leq \delta_s<1$, such that the R.I.P. condition holds. It stems directly that the smallest singular value $\sigma_{min}$ of the matrix $\mathbf{\Phi}_{S_{ac}}=[\mathbf{Q}\ \mathbf{I}_S]$ is lower bounded by 
\begin{equation}
\label{eq:min_sin_val}
\sigma_{min}\geq \sqrt{1-\delta_s}.
\end{equation}
\end{lemma}
\begin{proof}
Let $\boldsymbol{v}_m$ be the eigenvector which is associated with the smallest singular value of $\mathbf{\Phi}_{S_{ac}}$, then 
\[
\left\| \mathbf{\Phi}_{S_{ac}} \boldsymbol{v}_m \right\|_2^2= \sigma_{min}^2 \left\| \boldsymbol{v}_m \right\|_2^2.
\]
Since \eqref{eq:rip} holds for every vector, \eqref{eq:min_sin_val} follows.
\end{proof}
\begin{theorem}
In the case where GARD recovers the exact support of the sparse outlier vector $\boldsymbol{u}_0$, it approximates the ideal solution $\boldsymbol{\theta}_0$, with estimate $\boldsymbol{\theta}_*,$
acquiring an error
\begin{equation}
\label{eq:error_bound}
|| \boldsymbol{\theta}_*-  \boldsymbol{\theta}_{0}||_2 \leq \frac{\epsilon_0}{\tau \sqrt{1-\delta_s}}. 
\end{equation}
where $\tau$ is the smallest singular value of matrix $\mathbf{X}.$
\label{theor:errorbound}
\end{theorem}
\begin{proof}
The proof follows the same concepts as the stability result of Theorem 5.1 in \cite{donoho2006stable}. Since the support set of the sparse estimate $\boldsymbol{u}_*$ is also $S$, matrix $\mathbf{\Phi}=[\mathbf{Q}\ \mathbf{I}_n]$ could also be written as $\mathbf{\Phi}=[\mathbf{\Phi}_{S_{ac}}\ \mathbf{I}_{S^c}]$. Thus, we have that 
\[
\boldsymbol{z}_*= \begin{pmatrix}
\boldsymbol{w}_{*}\\ [\boldsymbol{u}_{*}]_S
\end{pmatrix}  := \argmin_{\boldsymbol{z}} ||\boldsymbol{y}-\mathbf{\Phi}_{S_{ac}}\boldsymbol{z}||_2^2=\mathbf{\Phi}_{S_{ac}}^{\dagger} \boldsymbol{y},
\]
where matrix $\mathbf{\Phi}_{S_{ac}}^{\dagger}$ denotes the Moore-Penrose pseudoinverse of matrix $\mathbf{\Phi}_{S_{ac}}$. Equation \eqref{eq:probdefmat_qr_noise} (which is equivalent to \eqref{eq:probdefmat2}) could also be written in a more compact form as 
\[
\boldsymbol{y}= \mathbf{\Phi}_{S_{ac}}\boldsymbol{z}_0 +\boldsymbol{\eta},\]
where $\boldsymbol{z}_{0}= \begin{pmatrix}
\boldsymbol{w}_{0}\\ [\boldsymbol{u}_{0}]_S
\end{pmatrix}$. 
Hence, we take
\[
\boldsymbol{z}_* = \mathbf{\Phi}_{S_{ac}}^{\dagger} \boldsymbol{y} = \boldsymbol{z}_0 + \mathbf{\Phi}_{S_{ac}}^{\dagger}\boldsymbol\eta.
\]
Finally,
\begin{align}
|| \boldsymbol{z}_*-  \boldsymbol{z}_{0}||_2 &\leq || \mathbf{\Phi}_{S_{ac}}^{\dagger} \boldsymbol\eta ||_2 \leq ||\mathbf{\Phi}_{S_{ac}}^{\dagger} ||_2\cdot  ||\boldsymbol\eta ||_2\nonumber\\
&\leq \sigma_{min}^{-1}\epsilon_0\leq \epsilon_0/\sqrt{1-\delta_s},
\label{eq:er_bound}
\end{align}
where we have also used that $|| \mathbf{\Phi}_{S_{ac}}^{\dagger}||_2$ is bounded, using the smaller singular value $\sigma_{min}$ of matrix $\mathbf{\Phi}_{S_{ac}}$, as well as \eqref{eq:min_sin_val}. The result follows from the fact that \[ \| \boldsymbol{\theta}_*  - \boldsymbol{\theta}_0 \|_2 \leq  \| \mathbf{R}^{-1} \|_2 \|  \boldsymbol{z}_* - \boldsymbol{z}_0 \|_2, \]
where $\|\mathbf{R}^{-1} \|_2$ is the spectral norm of $\mathbf{R}^{-1}$ equal to $\sigma_{\min}(\mathbf{R})^{-1}$. Since $\mathbf{X}=\mathbf{QR}$, the smallest singular value of $\mathbf{R}$ equals\footnote{Matrices $\mathbf{X}$ and $\mathbf{R}$ share the same singular values.} the smallest singular value $\tau=\sigma_{\min}(\mathbf{X})$ of $\mathbf{X}$, thus the proof is complete.
\end{proof}

\begin{remark}
Let 
\[
c = \sqrt{\frac{\min \{|u_i|,\; u_i\not=0\}-(2+\sqrt{6})\epsilon_0}{2  ||\boldsymbol{u}_0||_2 }}.
\]
Although, $c$ is readily computed, recall that $\delta_s$ is not, since this constant encloses the combinatorial nature of the problem for all the possible subsets of cardinality at most $s$. As a consequence, inequalities \eqref{eq:delta_bound_definition}, \eqref{eq:delta_bound_definition_epsilon}, \eqref{eq:min_sin_val} and \eqref{eq:error_bound}, could not be verified in practice; nonetheless, they all serve significant theoretical purposes. 
\end{remark}
\begin{remark}
Combining \eqref{eq:error_bound} with \eqref{eq:delta_bound_definition_epsilon}, we also have the following bound for the approximation of $\boldsymbol{\theta}_0$:
\begin{equation}
\label{eq:ds_approx_er}
\left\| \boldsymbol{\theta}_{*} - \boldsymbol{\theta}_0 \right\|_2 \leq \frac{\epsilon_0}{\tau \sqrt{1-c}},
\end{equation}
which due to its immediacy will be tested and verified later, in section \ref{sec:exper}. However, it is looser than that of \eqref{eq:error_bound}.
\end{remark}
\begin{remark}
\label{rem:bound_of_outl_e0}

The bound $c$ in \eqref{eq:delta_bound_definition_epsilon}, clearly depends on the sparsity level and values of the  outlier vector, but also on the inlier noise bound $\epsilon_0$. Also notice, that $\epsilon_0=0$ leads to the bound of $\delta_s$ in the noiseless case, i.e., \eqref{eq:delta_bound_definition}. Since in \eqref{eq:delta_bound_definition_epsilon} two terms affect the bound, we cannot except to recover the support exactly, in case both outlier noise of high density and heavy inlier noise exist. Such a scenario, would imply the bound on $\delta_s$ to be extremely tight and it is likely that it could not be satisfied. Finally, notice that $ \min \{|u_i|,\; u_i\not=0 \}$ should be greater than $(2+\sqrt{6})\epsilon_0$, if we would like \eqref{eq:delta_bound_definition_epsilon} to be valid. This could also be considered as a measure, of when an error value is considered as an outlier.
\end{remark}


\section{Experiments}
\label{sec:exper}

\begin{table*}
\begin{center}
\begin{tabular}{|c|c|c|}
\hline
Algorithm&Sol. to problem&Complexity \\
\hline
(GARD) &\eqref{eq:probformL0} or \eqref{eq:probformL0tog}  & $O\big(m^3/3 + k^3/2 +(n+3k)m^2 + 3kmn \big),\ k<< n $\\ \hline
(M-est) & \eqref{eq:M-est} &  $O\big( m^3/3+ nm^2 \big)$/step\\ \hline
(SOCP) &\eqref{eq:SOCP}  & $O\big( (n+m)^{2.5}n\big)$\\ \hline
(ADMM) &\eqref{eq:genlasso}  & $O\big((n+m)^3/3 + n(n+m)^2 \big)$/step\\
\hline
(SBL) &\eqref{eq:Bayes1}  & $O\big(m^3/3 +nm^2 \big)$/step\\
\hline
(ROMP) &  \eqref{eq:ROMP} & - \\ \hline
\end{tabular}
\end{center}
\caption{For (GARD), $k$ is the number of times the algorithm identifies an outlier. For (GARD) and (SOCP) total complexity is given. For the rest total complexity depends on the number of iterations for convergence.}
\label{tab:complexity}
\end{table*}

The setup for each one of the pre-existing methods (see section \ref{prel&relwork}), which are compared with GARD, are listed bellow:
\begin{itemize}
\item (M-est): In the following experiments, Tukey's biweight (or bisquare) robust (but nonconvex) function has been employed. This is included in the MATLAB function ``robustfit''. For $\sigma$, we have used the default parameter value setting (see \cite{rousseeuw2005robust, maronna2006robust}).

\item (SOCP): In order to solve the (SOCP) problem, we employed the MATLAB function ``SeDuMi'', which is included in the optimization package ``CVX''  of Stanford University, (CVX RESEARCH: http://cvxr.com/ (31/01/2014)). The input parameter for SeDuMi, is the bound of the inlier noise, used  for the definition of the second order cone.

\item (ADMM): For this method, parameter $\lambda$ should be predefined. Furthermore, the parameter, $\rho$, that is used for the soft thresholding operator is also predefined (low) to $\rho_{0}=10^{-4}$ and adapts at each step via $\rho_{i}=min\{5, 1.1\rho_{i-1}\}$. We have also employed a termination criterion, when the norm of the estimate undergoes changes from one step to the next, less than the predefined threshold of $10^{-4}.$

\item (SBL): Input parameters, $\sigma_{(0)}^2$, $\boldsymbol\theta_{(0)}$  and $\gamma_{i_{(0)}}$ are initialized. Following \cite{tipping2001sparse, wipf2004sparse}, we have also pruned the hyperparameters $\gamma_{i_{(k)}}$ from future iterations, if they become smaller than a predefined threshold (set low to $10^{-5}$).  Although the computational cost for Robust (SBL) is $O(m^3/3 +nm^2 )$ per step, the total cost depends on other variables too; such are the number of hyperparameters, that are pruned from future iterations, as well as the number of iterations needed for convergence. This is also the case for other methods, too.

\item (ROMP): The algorithm makes use of OMP's main iteration loop; in the first iteration, just a single column from matrix $\mathbf{X}$ participates in the M-est solution and each time the number of columns is augmented by one. Since the method solves an (IRLS) (or M-est at each step), instead of solving a Least-Squares problem restricted on the active set of columns, the complexity of the algorithm is not given in closed form. Once again, we have used Tukey's biweight function (\textit{robustfit}) with the default parameter settings, as in the M-est. The algorithm is terminated once the residual error drops below the bound of the inlier noise $\epsilon_0$.
\end{itemize}

In the experiments section, we have tested and analyzed the performance of each related algorithm. The experimental set up parallels that of \cite{mitra2013analysis}. Our data $(y_i,\boldsymbol{x}_i),\ i=1,2,...,n$, $\boldsymbol{x}_i\in \Real^m$ are generated via equation \eqref{eq:probdefmat2}; for the case where no inlier noise exists, we have set $\boldsymbol{\eta}=\boldsymbol{0}$. The rows of matrix $\mathbf{X}$, i.e., $\boldsymbol{x}_i$'s, are obtained by uniformly sampling an $m$-dimensional hypercube centered around the origin and $\boldsymbol{\theta}_0 \in \Real^m$ are random vectors, with values chosen from the normal distribution with mean value $0$ and standard deviation set to $5$.

\begin{figure}
\begin{subfigure}{.5\textwidth}
  \centering
  \includegraphics[scale=0.5]{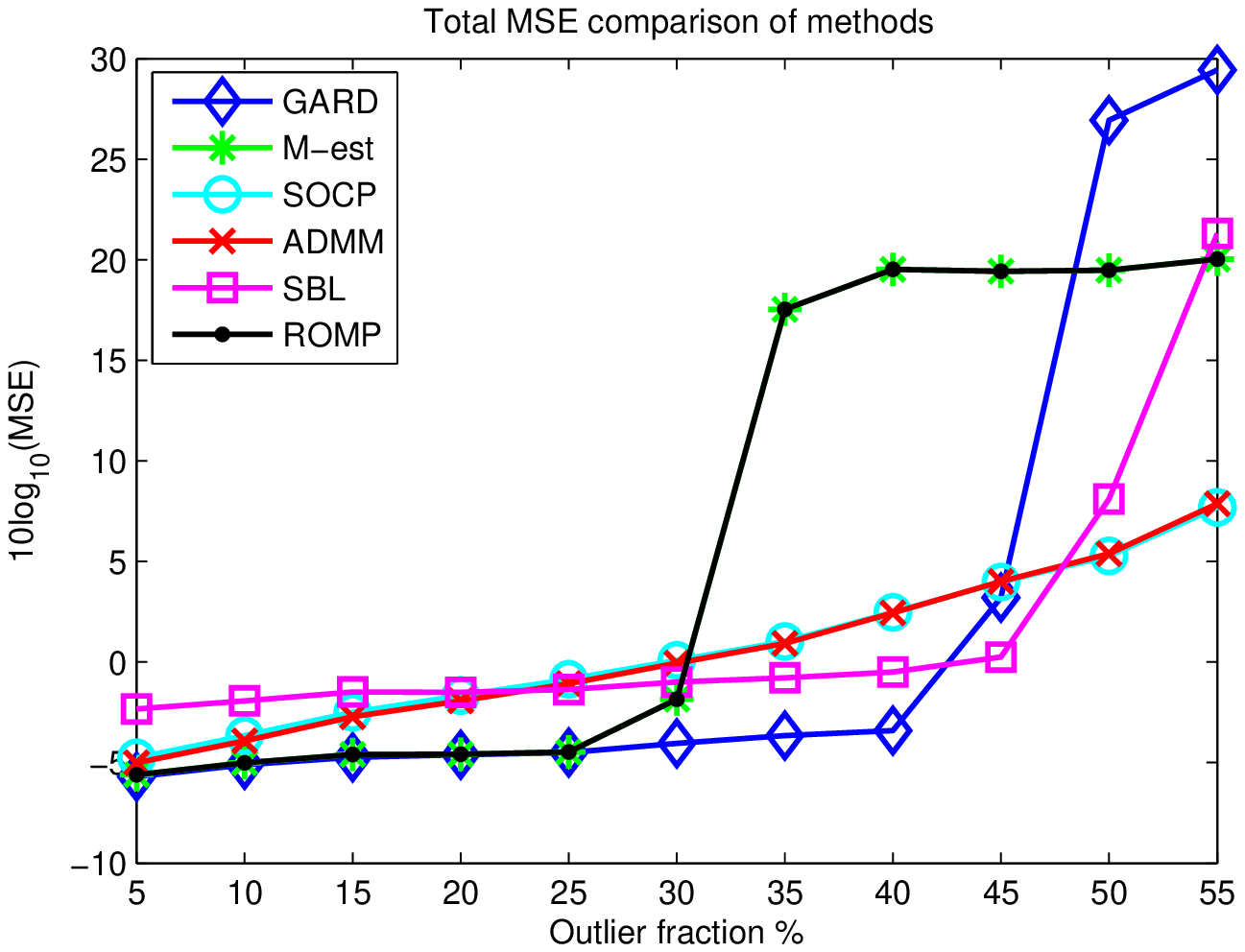}
  \caption{Dimension of the unknown vector $m=50$.}
  \label{fig:sfig1}
\end{subfigure}%
\begin{subfigure}{.5\textwidth}
  \centering
  \includegraphics[scale=0.5]{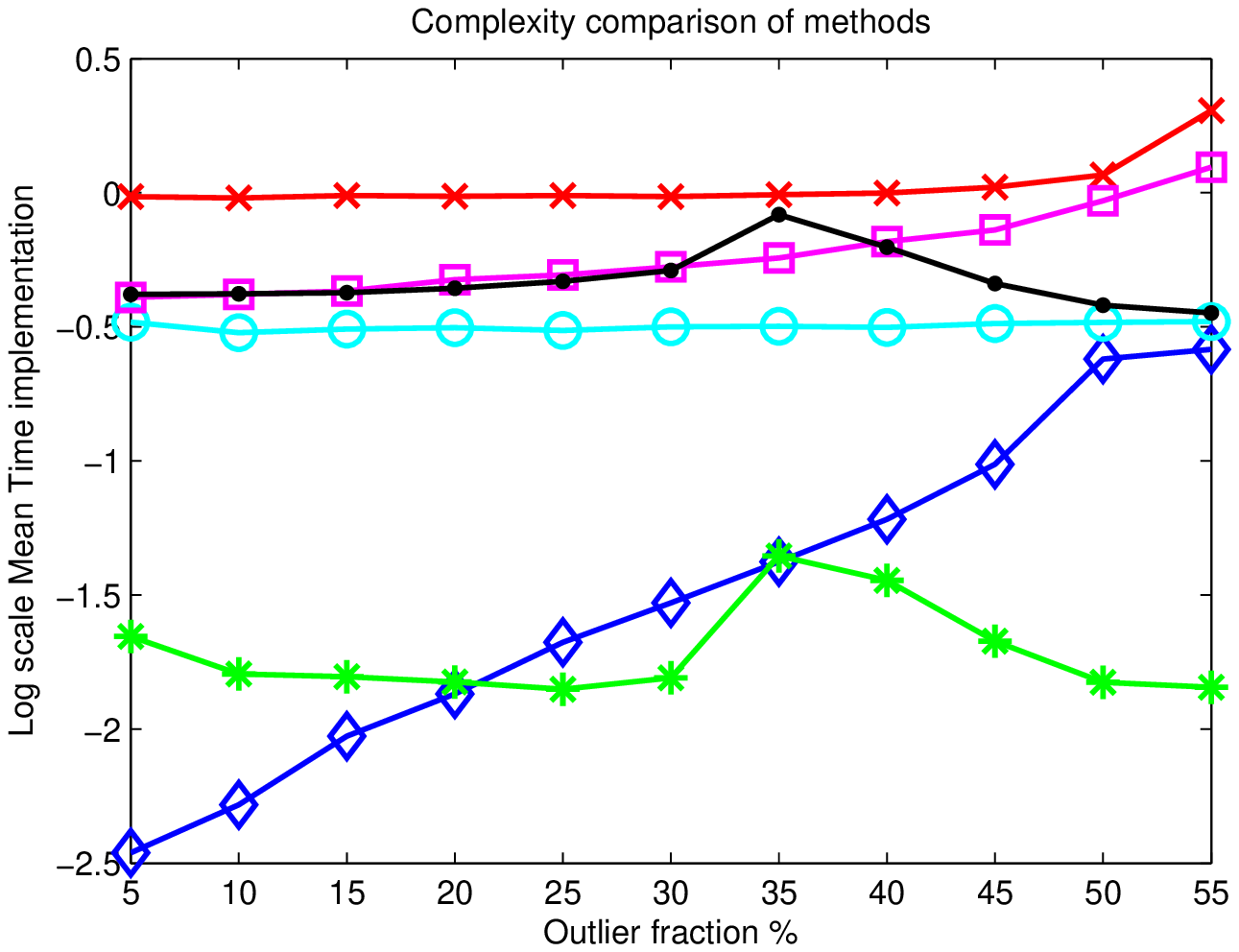}
  \caption{Dimension of the unknown vector $m=50$.}
  \label{fig:sfig1}
\end{subfigure} \quad   
\begin{subfigure}{.5\textwidth}
  \centering
  \includegraphics[scale=0.5]{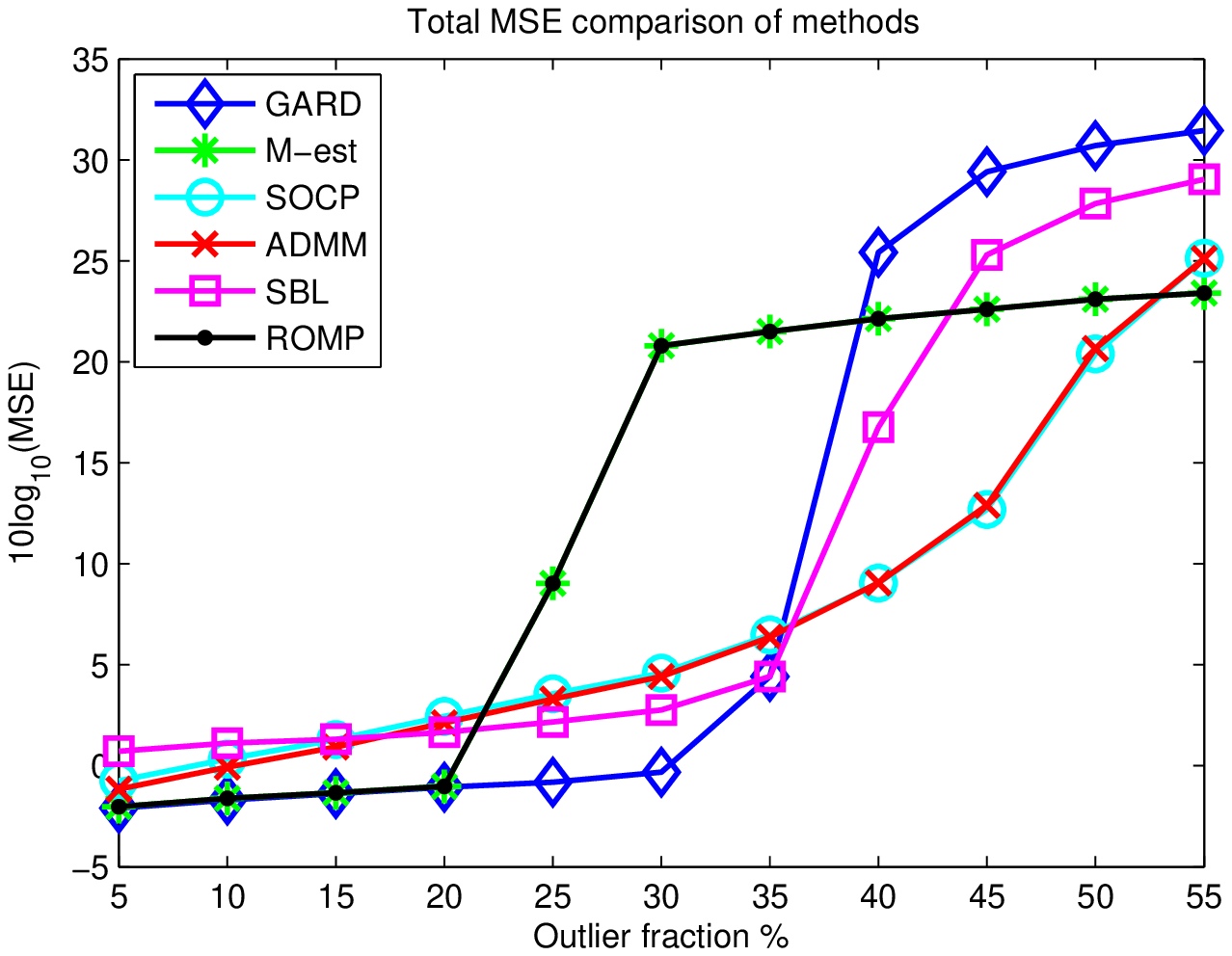}
  \caption{Dimension of the unknown vector $m=100$.}
  \label{fig:sfig1}
\end{subfigure}%
\begin{subfigure}{.5\textwidth}
  \centering
  \includegraphics[scale=0.5]{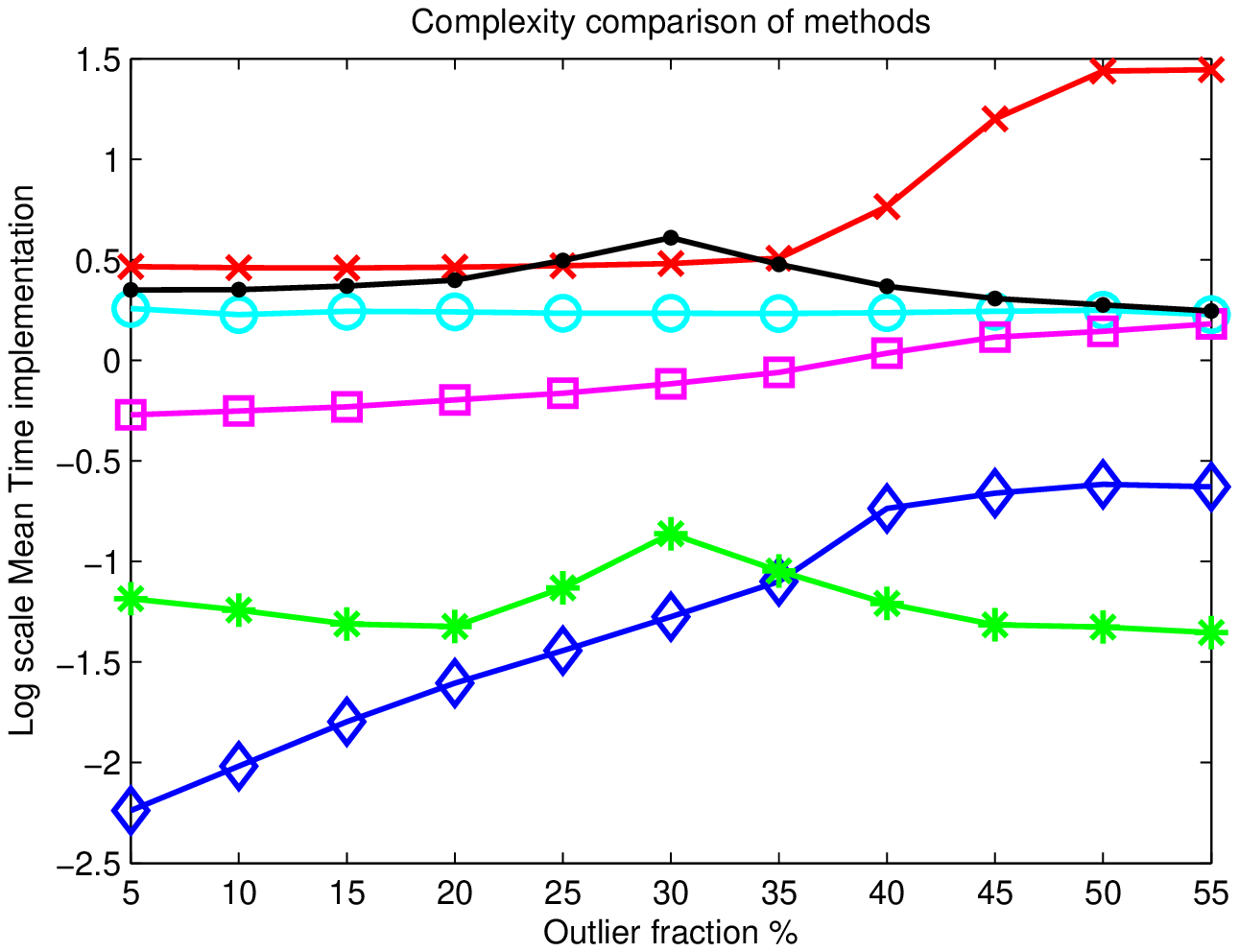}
  \caption{Dimension of the unknown vector $m=100$.}
  \label{fig:sfig2}
\end{subfigure}  \quad   
\begin{subfigure}{.5\textwidth}
  \centering
  \includegraphics[scale=0.5]{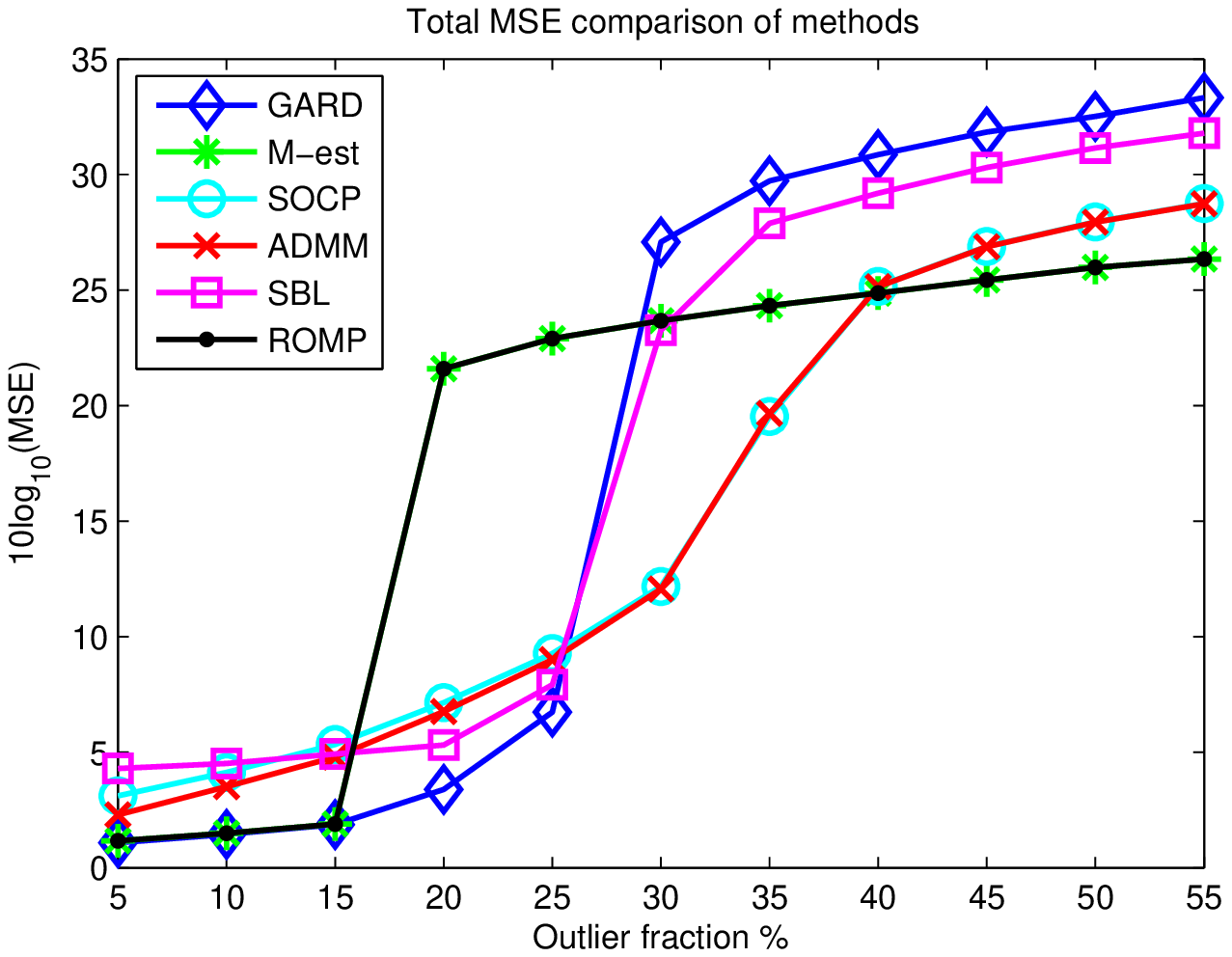}
  \caption{Dimension of the unknown vector $m=170$.}
  \label{fig:sfig2}
\end{subfigure}
\begin{subfigure}{.5\textwidth}
  \centering
  \includegraphics[scale=0.5]{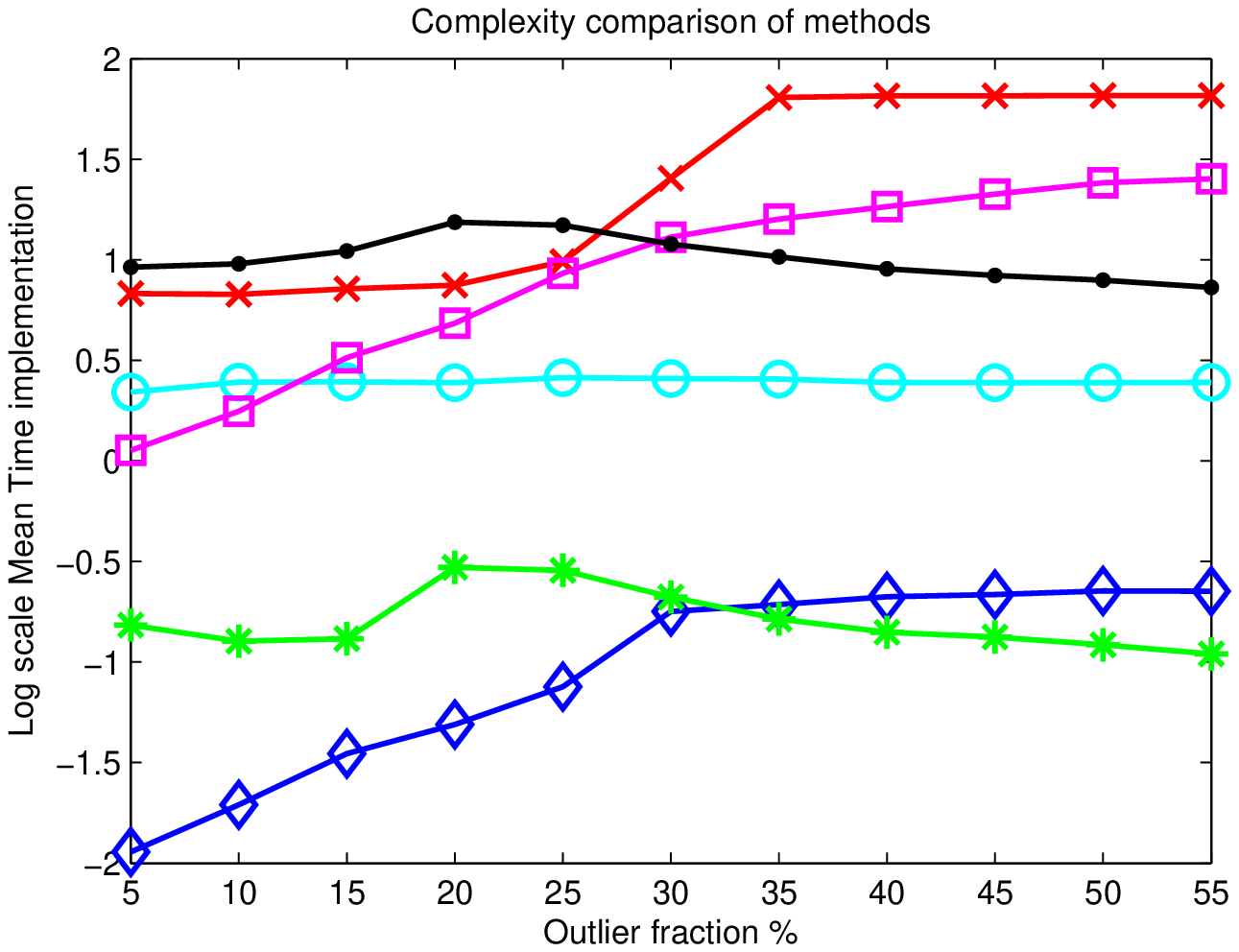}
  \caption{Dimension of the unknown vector $m=170$.}
  \label{fig:sfig2}
\end{subfigure}
\caption{Figures (a), (c), (e): The attained MSE versus the outlier fraction, for various dimensions of the unknown vector $\boldsymbol\theta_{0}$.  In all cases, $n=600$ observations were used. Figures (b), (d), (f): Log-scale of Mean Time versus the fraction of outliers.}
\label{fig:MSE-time}
\end{figure}

\subsection{Mean-Square Error estimation}
\label{sec:MSEest}

In the first experiment, we compared all methods with respect to the mean-square error (MSE), which is computed as the average, over  $100$ realizations at each outlier vector density, of the squared norm of the difference between the estimation vector $\boldsymbol{\theta}_*$ and the unknown vector $\boldsymbol{\theta}_0$. In parallel, we have also computed the average time (MT) (in sec) that each method requires to complete the estimation, for each outlier density. Aiming for more details, we have plotted the results in a logarithmic scale for each dimension of the unknown vector/signal $\boldsymbol\theta_{0}$  ($m=50,100,170$).

Outlier values are equal to $\pm 25$, in $s$ indices, uniformly sampled over $n$ coordinates ($s<n$). Although, an outlier vector is sparse by definition, 
in some experiments we extended the density level, in order to test each method to its limits. The inlier noise vector has elements drawn from the standard Gaussian distribution, with $\sigma = 1$ and inlier noise bound $\epsilon_0$, which is assumed to be known.

The input parameter for GARD, SOCP and ROMP is the inlier noise bound $\epsilon_0$. For ADMM, the regularization parameter is set to $\lambda=1.2$. Note that all methods were carefully tuned so that to optimize their performance. A major drawback of the SBL is its sensitivity to the choice of the initial values. Recall that this is a non-convex method, which cannot guarantee that the global minimum is attained for each dimension $m$, while the time needed for each implementation cannot be assured, since the number of iterations until convergence strongly depend on those parameters. Hence, for this method, random initialization was performed a number of times and the best solution was selected. Finally, the (M-est) does not require any predefined parameters.

Figure \ref{fig:MSE} shows the MSE (in dBs), versus the fraction of the sparse outlier vector for various dimensions of the unkown vector $\boldsymbol{\theta}_0$. The MT, that is required for each algorithm to converge is shown in Figure \ref{fig:time} in logarithmic scale. Although complexity of each method is already addressed (table \ref{tab:complexity}), in certain algorithms, the number of iterations until convergence greatly influences the required total implementation time. Observe that GARD attains the lowest MSE among the competitive methods for outlier fraction lower than $40\%,\ 35 \%$ and $25\%$ for dimension of the unknown vector $m=50,100,170$, respectively. The performance of M-est and ROMP is also notable, since both methods also attain a low MSE. However, this is only possible for outlier fraction of less than $25\%,\ 20 \%$ and $15\%$ (MSE equal to that of GARD). In particular, we found that M-est and ROMP have identical performance, despite the fact that ROMP combines two methods, resulting to a higher computational cost.

It should also be noted, that in Figure \ref{fig:MSE} (b) and (c), all algorithms break their performance at lower outlier fractions with respect to GARD. However, the interesting \textit{zone} of outlier vector density, in real time applications, is between $0\% - 20 \%$ of the sample set, since greater percentages do not imply outlying values. Hence, GARD attains the lowest MSE within this sensitive zone.
Finally, the experiments show that ADMM and SOCP attain a similar performance, as expected, due to the fact that they both address the same problem. 

Besides its superior performance with respect to recovery error, GARD's computational requirements remain low. As shown in Figure \ref{fig:time}, GARD  appears to have the lower computational cost among its competitors, for outlier fraction less than $20\%$.

\subsection{Complexity evaluation for large data sets}
\label{sec:larg_deimension}
In the current section, we have attempted to evaluate the performance of the most computationally efficient methods, in the case where the number of generate data grows significantly, compared to the dimension of the unknown vector $\boldsymbol{\theta}_0$. Comparison is performed for all methods except from ADMM and ROMP. As presented in table \ref{tab:complexity}, the ADMM algorithm does not handle efficiently large numbers of samples. On the other hand, although ROMP performs exactly as M-est, that comes at a higher computational cost, therefore seemed impractical to put it to test. 
\begin{figure}
\begin{center}
\includegraphics[scale=0.6]{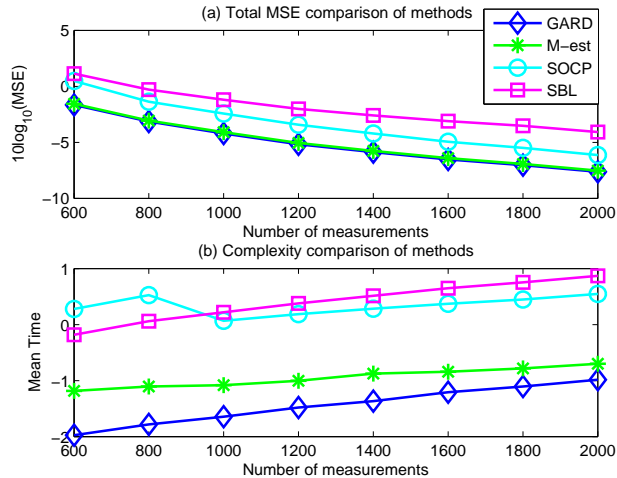}
\caption{Large scale complexity test for dimension of the unknown vector set at $m=100$. While varying the number of measurements, the MSE (a) and the Mean time until convergence (b), is shown for each method. It is clear that GARD attains the lowest MSE, whilst being the most efficient.}
\label{fig:large_data_sets}
\end{center}
\end{figure}
Once again, equation \eqref{eq:probdefmat2} has been used to generate our data. The dimension of $\boldsymbol{\theta}_0$ is set at $m=100$, the density of the outlier noise is $10 \%$  with values $\pm 25$ spread uniformly over $n$ coordinates and finally the inlier noise vector has elements  drawn from the standard Gaussian distribution, with $\sigma=1$ and inlier noise bound $\epsilon_0$, assumed to be known. For each number of measurements $n$, $100$ independent experiments have been performed, while varying the inlier and outlier noise; results have been averaged.

In figure \ref{fig:large_data_sets} (b), we have evaluated the mean implementation time (in logarithmic scale for precision) for each method, while at the same time the total MSE is measured, for each varying size of $n$. It is clear, that even for significantly large $n$, i.e., large number of measurements, GARD excels. Whilst attaining the lowest MSE, the mean time until convergence, is the lowest.

\subsection{Support recovery test}
\label{sec:supporttest}
This section attends to bridge the gap, between the theoretical results properties in section \ref{sec:theoretic} and the experimental performance of GARD. The results of section \ref{sec:MSEest}, showcase the performance of GARD. However, it would be incorrect to conclude that the support of the sparse outlier vector is correctly recovered, in cases where the algorithm attains a low MSE, a matter that we would like to address here. Although, the recovery of the sparse outlier support is desirable, since it guarantees the smallest MSE possible, it should be noted that GARD performs well (with respect to the MSE), even in cases where the recovery of the support is not exact; e.g., one of the most common cases is to identify a few extra indices (that do not belong to the support of $\boldsymbol{u}_0$) as outlying elements.\\
For all support recovery tests, we have set the dimension of the unknown vector $\boldsymbol{\theta}_0$, at $m=100$ and corrupted the original data with outliers in $s<n$ indices, uniformly sampled over $n=600$ measurements. Also, for each fraction of outliers, i.e.,  $(s/n)\cdot 100\%$, we performed $10000$ Monte Carlo runs.\\
Let $S_k$ denote the support set of the sparse estimate $\boldsymbol{u}_*$ and $S$ the support set of the sparse outlier vector $\boldsymbol{u}_0$. The green line corresponds to the percentage of correct indices the proposed scheme has recovered, i.e., indices $i \in S_k \subseteq S$, while the orange line corresponds to the extra indices that the method has identified as outliers, i.e., indices $j \in J\setminus S$. In parallel, since the smallest principal angle cannot be computed directly, we have tracked the bound of $c>\delta_s$ for evaluation of the theoretical results proposed in section \ref{subsec:onlyoutliers}. The vertical line, corresponds to the largest outlier fraction, that the proposed scheme succeeds in recovering the sparse outlier vector support, one to one.
\subsubsection{The presence of outliers only}
\label{subsub:only_outliers}
The scenario in which our original data is corrupted with outlier values only, is treated separately. Our data are generated via equation \eqref{eq:probdefmat2}, for $\boldsymbol{
\eta}=\boldsymbol{0}$ and outlier values\footnote{In the noiseless case, arbitrarily small values, are treated as outliers; thus the performance of GARD is not affected by a particular selection of those values.} $\pm 25$, in $s$ indices, uniformly sampled over $n$ coordinates.
\begin{figure}
\begin{center}
\includegraphics[scale=0.6]{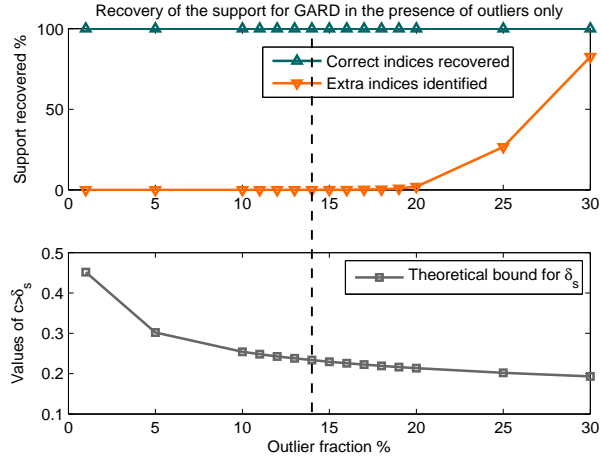}
\caption{Recovery of the support and relation to the bound of $\delta_s$, for the noiseless case. For outlier fraction of less than $14\%$, the bound for $\delta_s$ \eqref{eq:delta_bound_definition}, is guaranteed, hence the recovery is exact.}
\label{fig:sup_rec_outliers_only}
\end{center}
\end{figure}
In figure \ref{fig:sup_rec_outliers_only}, the recovery of the exact support versus fraction of outliers is demonstrated. It is clear that for fraction of less than $14\%$, the bound for $\delta_s$ as Theorem \ref{theor:exact_rec} proposes, is guaranteed, thus the recovery of the support is exact and also the approximation of $\boldsymbol{\theta}_0$ is of zero error. Noteworthy is also the fact that the approximation error is also zero, in cases where only a few extra indices that belong to $S^{c}$, are imported into the support set $S_k$. 
\subsubsection{The presence of both inlier and outlier noise}
\label{subsub:inlier&outliers}
In the current section, we have worked towards the empirical validation of \eqref{eq:delta_bound_definition_epsilon} and \eqref{eq:ds_approx_er}, where two separate tests have been performed. Equation \eqref{eq:probdefmat2}, has been also employed here to generate our input data. 

In the first test, we have fixed the maximum bound for the norm of the inlier noise vector at $\epsilon_0=28$, while we have altered the fraction of outliers. In order to achieve this, we have used Matlab's random generator for the Gaussian distribution with standard deviation depending on $\epsilon_0$, while we have cut off the largest elements (in the absolute sense) when it was required, so that the norm of the inlier noise vector always remains bounded by $\epsilon_0.$ Also, recall on remark \ref{rem:bound_of_outl_e0}, that the minimum element of the absolute value of the outlier vector should be grater that $(2+\sqrt{6})\epsilon_0$, in order \eqref{eq:delta_bound_definition_epsilon} to be valid. Thus, outlier values have been set at $\pm150$, while the values of $\boldsymbol{y}_0=\mathbf{X}\boldsymbol{\theta}_0$, range at $170-180$.

In figure \ref{fig:supportrec_noisy}, we have plotted the recovery of the support for GARD and its relation to the bound $c$ of the smallest principal angle $\delta_s$, for each outlier fraction. As one could observe, for fraction of outliers less than $13\%$, the bound for $\delta_s$ as Theorem \ref{theor:exact_support} proposes is guaranteed, thus the recovery of the support is exact. In parallel, we have computed the MSE between $\boldsymbol{\theta}_0$ and $\boldsymbol{\theta}_*$ and tracked the relation to the theoretical bound\footnote{Since the MSE is a squared norm between $\boldsymbol{\theta}$ and $\boldsymbol{\theta}_*$, the bound is the squared right hand side of \eqref{eq:ds_approx_er}.} of \eqref{eq:ds_approx_er}.

In the second test, the ability of GARD to deal with heavy loads of noise, is demonstrated. The outlier values were set at $\pm150$ and the bound of $\epsilon_0$ was increased, so that the inlier noise corresponds to noise of 20 dB. In such a case,  the bounds established in \eqref{eq:delta_bound_definition_epsilon} and \eqref{eq:ds_approx_er} are violated, however GARD manages to do well. In figure \ref{fig:support_rec_er2}, the recovery of the support versus the outlier fraction is demonstrated. We conclude, that although the method does not succeed to recover the sparse outlier support $100\%$, the MSE is relatively low, at least for low fraction of outliers, i.e., below $10\%$. It should be noted that the MSE value close to 5 is not high, compared to the MSE measured in figure \ref{fig:supportrec_noisy}, which was close to 1.
\begin{figure}
\begin{center}
\includegraphics[scale=0.6]{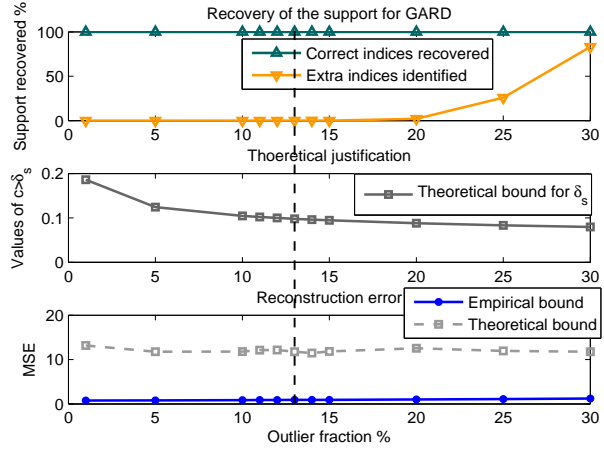}\caption{Recovery of the support and relation to the bound of $\delta_s$, for the case inlier and outlier noise coexist. For outlier fraction of less than $13\%$, the bound for $\delta_s$ \eqref{eq:delta_bound_definition}, is guaranteed, hence the recovery of the support is exact, while the MSE computed is valid under the bound that inequality \eqref{eq:ds_approx_er} suggests.}
\label{fig:supportrec_noisy}
\end{center}
\end{figure}

\begin{figure}
\begin{center}
\includegraphics[scale=0.6]{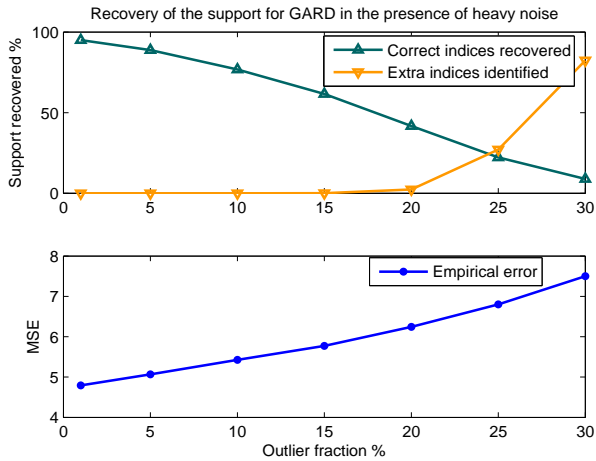}\caption{Recovery of the support for GARD, in the case where outlier and heavy inlier noise of approx. 20 dB coexist. Although, the support is not entirely recovered, the MSE is relatively low.}
\label{fig:support_rec_er2}
\end{center}
\end{figure}

\subsection{Phase transition curves}
\label{sec:phasetrans}
In the final experiment, we have tested whether each method succeeds to reach an estimate or not. In particular, we computed the probability for which each method succeeds in recovering the solution. Since ADMM and ROMP have a higher computational cost, as well as similar to other methods performance, we limited our efforts to test the other three methods, i.e., GARD, SOCP, M-est and SBL.

Figure \ref{fig:phasetrans} (a) shows the probability of recovery for each method tested, while varying the fraction of outliers. The dimension of the unknown vector $\boldsymbol\theta_{0}$ is $m=100$. For each density of the sparse outlier vector, we have computed the probability over $200$ Monte Carlo runs. For each method, we have assumed that the solution is found, if $|| \boldsymbol\theta_{*} - \boldsymbol{\theta}_0 ||_2/ ||\boldsymbol\theta_{0}||_2 \leq 0.03$. The major result, is that for fraction of outliers under $25\%$, GARD succeeds in recovering the solution, with probability 1. For (M-est), the percentage is below $20\%$, while for the rest $\ell_1$ minimization methods the percentage is even lower. For SBL, the probability to recover the solution is not guaranteed, even for the lowest fractions of outliers.

\begin{figure}
\centering
\begin{subfigure}[b]{0.5\textwidth}
  \includegraphics[scale=0.5]{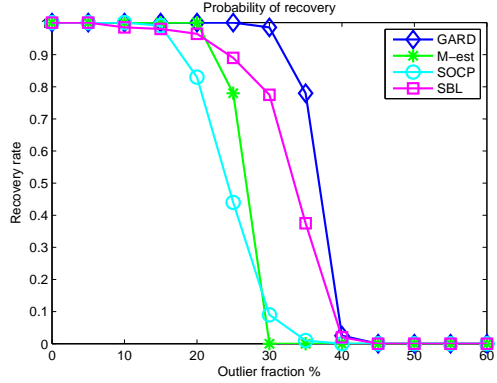}
\caption{Probability of recovery for dimension of the unknown vector $m=100$.}
\end{subfigure}
\quad
\begin{subfigure}[b]{0.5\textwidth}
  \includegraphics[scale=0.5]{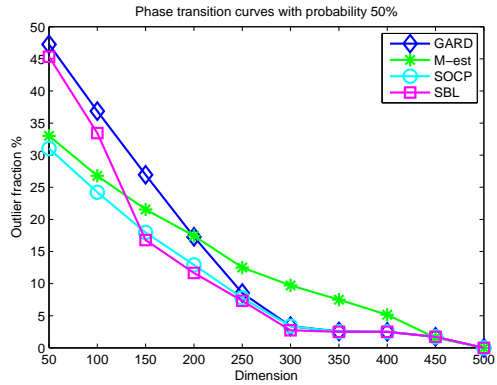}\caption{Phase Transition curves.}
\end{subfigure}
\caption{(a) Probability of recovery while varying the fraction of outliers. (b) Transition from success to failure of each method with probability $50\%$.}
\label{fig:phasetrans}
\end{figure}
Figure \ref{fig:phasetrans} (b), shows the phase transition curves for each method. For each dimension of the unknown vector $\boldsymbol\theta_{0}$, we have computed the fraction of outliers for which the method transits from success to failure with probability $50\%$. Experiments were carried over $200$ Monte Carlo runs. Once again, we have assumed that the solution is found, using the criterion as in Figure \ref{fig:phasetrans} (a). We can observe that for each fixed dimension of the unknown vector the probability for each method to recover the solution (always within a given tolerance), increases for fraction of outliers below each phase transition curve, while the probability decreases as we move above the phase transition curves. Here, it is clear that up to $m=200$, GARD succeeds to recover the solution with the highest probability from all the rest of the methods. However, for greater dimension of the unknown, the measurement dimension (here $n=600$) seems pretty ``poor" to allow GARD to preserve the performance (in the sense that more data is required), although it does not drop below the $\ell_1$ minimization techniques.

%

\subsection{Experiments with general noise forms}
\label{sec:pdfs}
In the current section we performed a set of more realistic experiments for the methods described and measured the MSE over an average of 100 Monte-Carlo runs. Equation \eqref{eq:probdefmat}, describes our model, where we produced ($n=600$) measurements corrupted with different types of noise and measured the MSE. The dimension of the unknown vector $\boldsymbol{\theta}_0$ is $m=100$. For all tests, the ADMM was excluded from the last set of experiments, since the method proved weak to handle different orders of noise values, thus failed to converge for all tests.
\begin{itemize}
\item \textbf{Tests A, B and C}. The noise vector is drawn from 
the L\' evy alpha-stable distribution, $\mathcal{S}(\alpha, \beta,\gamma,\delta)$, with pdf expressed in closed form only for special cases of the parameters defined. The distribution's parameters $\beta$ and $\delta$ that control symmetry, were set to zero (results to a symmetric distribution without skewness) for all three experiments. For A, the distribution's parameters were set to $\alpha=0.45$ and $\gamma=0.3$; the parameters for each method were set to $\epsilon_0=3$ for GARD and SOCP, $\sigma=1.2$ for M-est and ROMP, while the hyperparameters for SBL were initialized to $10^{-4}.$ In table \ref{tab:4test}, it can be seen that almost all methods perform quite well (MSE is low), with GARD appearing to perform better. For test B, $\alpha=0.4$, $\gamma=0.1$; for GARD $\epsilon_0=3$, for SOCP $\epsilon_0=2$, for M-est and ROMP $\sigma=1$, while for SBL the hyperparameters were initialized at random (Gaussian) with variance equal to $10^{-5}, $ although fails to converge, for all values of the paramaters tested. Once again, it can be readily seen that GARD attains the lowest MSE. Finally, for experiment C, $\alpha=0.3$, $\gamma=0.1$, resulting to more large values of noise; for GARD $\epsilon_0=3$, for SOCP $\epsilon_0=2$, for M-est and ROMP $\sigma=1$, while for SBL the hyperparameters were initialized at random (Gaussian) with variance equal to $10^{-6}$. In table \ref{tab:4test}, the MSE for GARD is significant lower than tests A and B, which means that the method identifies correctly the outlier values, regardless how large those are; the rest of the methods fail to provide a descent estimate for the unknown vector.

\item \textbf{Test D}. The noise consists of a sum of two vectors, drawn from 2 independent Gaussian distributions $\mathcal{N}(0,0.6^2)$ and $\mathcal{N}(0,0.8^2)$, plus an outlier noise vector of $10\%$ density (indices chosen uniformly at each repetition) with values $\pm 25$. The parameters required for each method are: the default tuning parameter for both M-est and ROMP; for GARD and SOCP $\max \{ \epsilon_1, \epsilon_2 \}$ is required, where $\epsilon_1,\epsilon_2$ are the bounds of each inlier noise vector, while for SBL an initialization at random with variance of $10^{-6}$ was performed. The model of the noise is now more complicated, hence the problem harder to solve for all methods. Once again, it is clear that GARD succeeds in handling this mixed type of noise too.
\end{itemize}
\begin{table}
\begin{center}
\begin{tabular}{|c|c|c|c|c|}
\hline
Algorithm & Test-A & Test-B & Test-C & Test-D \\ \hline
GARD &  0.1772  &  0.0180  & 0.0586   &  0.690  \\ \hline
M-est& 0.2248  &  0.2859  &  1.844e+06  &  0.704  \\ \hline
SOCP &  0.4990  & 0.3502   & 5.852e+05   & 1.011   \\ \hline
SBL  &  0.9859  & 58.3489   & 2.165e+06   & 1.292   \\ \hline
ROMP & 0.2248   & 0.2859   &  1.844e+06  &  0.704   \\ \hline
\end{tabular}
\end{center}
\caption{Computed MSE, for various experiments. In tests A,B and C, the noise is drawn from the Heavy-tailed distribution alpha-stable of Levy distribution. In test D, noise consists of a sum of two vectors, drawn from 2 independent Gaussian distributions with different variance, plus an outlier noise vector of impulsive noise.}
\label{tab:4test}
\end{table}

\section{Conclusions}\label{SEC:CONCL}
A novel algorithmic scheme, i.e., GARD, for robust linear regression has been developed. GARD alternates between an OMP selection step, which identifies the outliers, and a Least-Squares estimator, that attempts to fit the data. Several properties regarding convergence, error bounds and uniqueness of the solution have been derived. Furthermore, more theoretical results concerning the stability of the method and the recovery of the outliers' support have been extracted.

The proposed scheme has been compared with other well established techniques through extensive simulations. The experiments suggest that GARD has an overall tolerance in outliers compared to its competitors. Moreover, it attains the lowest error for the estimation of the unknown vector, along with M-est and ROMP; moreover, GARD attains similar MSE at lower complexity.

Finally, the experiments verify that our greedy-based GARD algorithm outperforms the $\ell_1$ norm-based  schemes, for low sparsity levels; since in practical applications outliers are expected to be just a few, greedy-based techniques seem to be the right choice.


%
\appendix 

\section{Proof of Theorem \ref{theor:exact_rec}}
\label{appendix:A}
Since matrix $\mathbf{X}$ is assumed to be full rank, according to the analysis presented in section \ref{subsec:onlyoutliers}, equation $\boldsymbol{y}=\mathbf{X}\boldsymbol{\theta}_0 + \boldsymbol{u}_0$ could be transformed into equation \eqref{eq:uniq_dec}. Thus, the principal angle defined in \eqref{eq:principalangle2}, is now involved.

Rather than delving into the main arguments of the proof, we need to establish the following propositions.
\begin{proposition}
Let $\mathbf{Q}$ be the orthonormal matrix of the reduced \textit{QR} decomposition of the full rank matrix $\mathbf{X}$ and $\delta_s$ the principal angle between the subspace spanned by $\SPAN(\mathbf{Q})$ and the subspace spanned by all the s-dimensional outlier subspaces. Then
\begin{equation}
||\mathbf{Q}^T \boldsymbol{v} ||_2 \leq \delta_s || \boldsymbol{v} ||_2
\label{eq:vec_bound}
\end{equation}
holds for every vector $\boldsymbol{v} \in \Real^n$ with $||\boldsymbol{v}||_0 \leq s$.
\label{prop:vec_bound}
\end{proposition}
\begin{proof}
The proof is straightforward by the definition of $\delta_s$ and equation \eqref{eq:principalangle3}:
\begin{align*}
\|\mathbf{Q}^T \boldsymbol{v} \|_2^2 &=| \langle \boldsymbol{v},\mathbf{QQ}^T \boldsymbol{v} \rangle |\leq \delta_s \|  \boldsymbol{v}\|_2  \| \mathbf{Q}\mathbf{Q}^T \boldsymbol{v}\|_2\\
 &\leq \delta_s \|  \boldsymbol{v}\|_2  \| \mathbf{Q}\|_2\|\mathbf{Q}^T \boldsymbol{v}\|_2 = \delta_s \|  \boldsymbol{v}\|_2  \|\mathbf{Q}^T \boldsymbol{v}\|_2,
\end{align*}
which leads to \eqref{eq:vec_bound}.
\end{proof}
\begin{lemma}
Let the assumptions of Proposition \ref{prop:vec_bound} be satisfied and $S$ be any non-empty subset of $J=1\textrm{---}n$ with cardinality $|S|=k\leq s<n$. Then
\begin{equation}
||\mathbf{Q}^T \mathbf{I}_S ||_2 \leq \delta_s
\label{eq:norm_bound}
\end{equation}
holds for every such set $S$.
\label{lem:norm_bound}
\end{lemma}
\begin{proof}
Let $\undertilde{\boldsymbol{v}} \neq \boldsymbol{0}$ be a vector of $\Real^k$, $k\leq s$. It is clear that $\mathbf{I}_S \undertilde{\boldsymbol{v}}= \boldsymbol{v} \in \Real^n$, with $|| \boldsymbol{v}||_0\leq s$ and $|| \undertilde{\boldsymbol{v}}||_2 = ||\boldsymbol{v} ||_2$.
Thus, for all $\undertilde{\boldsymbol{v}}\in\Real^k$ we have
\[
||\mathbf{Q}^T\mathbf{I}_S \undertilde{\boldsymbol{v}}||_2 = ||\mathbf{Q}^T \boldsymbol{v}||_2 \leq \delta_s ||\undertilde{\boldsymbol{v}}||_2,
\]
due to Proposition \ref{prop:vec_bound}.
The result follows from the definition of the matrix $2-$norm.
\end{proof}
The importance of Lemma \ref{lem:norm_bound} is twofold. First of all, it is a bound on the 2-norm of the matrix $\mathbf{Q}^T \mathbf{I}_S $. Moreover, since $||\mathbf{I}_S^T\mathbf{Q} \mathbf{Q}^T \mathbf{I}_S ||_2=|| \mathbf{Q}^T \mathbf{I}_S||_2^2$ and assuming that \eqref{eq:delta_bound_definition} holds, we have that
\begin{equation}
||\mathbf{I}_S^T\mathbf{Q} \mathbf{Q}^T \mathbf{I}_S ||_2\leq \delta_s^2 <1/2,
\label{in:usual_bound_byhalf}
\end{equation}
which leads to the fact that matrix $\mathbf{M}_k=\mathbf{I}_k-\mathbf{I}_S^T\mathbf{Q} \mathbf{Q}^T \mathbf{I}_S$ is invertible for all $S$ with $|S|=k\leq s$ and
\begin{equation}
|| \mathbf{M}_k^{-1}||_2\leq (1-||\mathbf{I}_S^T\mathbf{Q} \mathbf{Q}^T \mathbf{I}_S ||_2)^{-1}<2,
\label{eq:norm<2}
\end{equation}
due to a very popular Proposition of linear algebra.

\begin{lemma}
Let the assumptions of Lemma \ref{lem:norm_bound} be satisfied. Then
\begin{equation}
||\mathbf{I}_{S}^T\mathbf{QQ}^T \boldsymbol{v} ||_2 \leq \delta_s^2 || \boldsymbol{v} ||_2
\label{eq:vec_bound_fin}
\end{equation}
holds for every vector $\boldsymbol{v} \in \Real^n$, with $||\boldsymbol{v} ||_0\leq s$.
\label{lem:vec_bound_fin}
\end{lemma}
\begin{proof}
The tricky part of the proof is that the s-sparse vector $\boldsymbol{v} \in \Real^n$, might not necessarily share the same support set $S$. Let $S'$ the support set of vector $\boldsymbol{v}$, with $|S'|=k \leq s$. Thus, using $\boldsymbol{v}_{S'} \in \Real^k$ to denote the non-sparse vector (also notice that $\|\boldsymbol{v}\|_2 =\| \boldsymbol{v}_{S'}\|_2$), we have $\boldsymbol{v}=\mathbf{I}_{S'}\boldsymbol{v}_{S'}$. Hence, due to the sub-multiplicative property of the matrix 2-norm, we have
\begin{align*}
\|\mathbf{I}_{S}^T\mathbf{QQ}^T \boldsymbol{v} \|_2  &=  \|\mathbf{I}_{S}^T\mathbf{QQ}^T \mathbf{I}_{S'}\boldsymbol{v}_{S'} \|_2\\
& \leq  \|\mathbf{I}_{S}^T\mathbf{Q}\|_2 \|\mathbf{Q}^T \mathbf{I}_{S'}\|_2\|\boldsymbol{v}_{S'}\|_2  \leq \delta_s^2 ||\boldsymbol{v} ||_2,
\end{align*}
where we have used $||\mathbf{I}_{S}^T\mathbf{Q}||_2=||\mathbf{Q}^T\mathbf{I}_{S}||_2$ and both the results of Proposition \ref{prop:vec_bound} and Lemma \ref{lem:norm_bound}.
\end{proof}

\noindent\textbf{Proof of the Main Theorem}
\begin{proof}
Let $S=\supp(\boldsymbol{u}_0)$. At the initial step of (GARD), the Least Squares solution over the active columns of matrix $\mathbf{\Phi}=[\mathbf{Q}\ \mathbf{I}_n]$ is computed, i.e., columns vectors of matrix $\mathbf{Q}$. The initial residual is $\boldsymbol{r}^{(0)}=\boldsymbol{y}-\mathbf{Q}\boldsymbol{w}_*^{(0)}.$ At this point, we could express the residual in terms of the projection matrix $\mathbf{P_Q}$ onto the range of matrix $\mathbf{Q}$ \footnote{Take into account that $\mathbf{Q}$ is orthonormal.}. Thus, $\mathbf{P_Q}=\mathbf{QQ}^{T}$ and the residual could be written as
\[
\boldsymbol{r}^{(0)}=(\mathbf{I}_n- \mathbf{P_Q})\boldsymbol{y}=(\mathbf{I}_n- \mathbf{P_Q})\boldsymbol{u}_0,
\]
as  $\boldsymbol{y}=\mathbf{Q}\boldsymbol{w}_0+\boldsymbol{u}_0$, $\mathbf{I}_n- \mathbf{P_Q}$ is the projector for the subspace complementary to that of $\mathcal{R}(\mathbf{Q})$ and  $(\mathbf{I}_n- \mathbf{P_Q})\mathbf{Q}\boldsymbol{w}_0=\mathbf{0}$.

At the first step, in order to ensure a selection from the correct support $S$, we impose that
\begin{equation}
|r^{(0)}(i)|>|r^{(0)}(j)|,\ \forall\ i \in S\ \text{and}\ j \in S^c.
\label{eq:guarantee}
\end{equation}
The basic concept of the proof is to obtain lower and upper bounds for the left and right part of equation \eqref{eq:guarantee}. Employing Lemma \ref{lem:vec_bound_fin}, the left part is bounded below by
\begin{align}
|&r^{(0)}(i)| =|  \langle \boldsymbol{r}^{(0)}, \boldsymbol{e}_i \rangle |= | \langle \boldsymbol{u}_0 - \mathbf{QQ}^{T}\boldsymbol{u}_0, \boldsymbol{e}_i \rangle  |\geq  \nonumber \\
&\geq |u_i | - |\langle \mathbf{QQ}^{T}\boldsymbol{u}_0, \boldsymbol{e}_i \rangle | = |u_i | - |\boldsymbol{e}_i^T \mathbf{QQ}^{T}\boldsymbol{u}_0 |\nonumber\\
&\geq \min_{l\in S}|u_l| - \delta_s^2 ||\boldsymbol{u}_0||_2,
\label{eq:guarantee_bound_left}
\end{align}
where $u_i$ are the elements of $\boldsymbol{u}_0$.

Following similar steps, the right part is upper bounded by
\begin{align}
|r^{(0)}(j)| &=|  \langle \boldsymbol{r}^{(0)}, \boldsymbol{e}_j \rangle |= | \langle \boldsymbol{u}_0 - \mathbf{QQ}^{T}\boldsymbol{u}_0, \boldsymbol{e}_j \rangle  |= \nonumber \\
&= |\boldsymbol{e}_j ^T \mathbf{QQ}^{T}\boldsymbol{u}_0,  | \leq \delta_s^2 ||\boldsymbol{u}_0||_2,
\label{eq:guarantee_bound_right}
\end{align}
using that $\langle \boldsymbol{u}_0, \boldsymbol{e}_j \rangle =0$, since $j\in S^c.$

Hence, if we impose
\[
\min_{l\in S}|u_l| - \delta_s^2 ||\boldsymbol{u}_0||_2 > \delta_s^2 ||\boldsymbol{u}_0||_2,
\]
condition \eqref{eq:guarantee} is guaranteed and one of the correct columns, i.e., $i_1$, is bound to be selected at the first step (note, that it is not guaranteed that the largest outlier value will be selected first).

Considering $S_1=\{j_1\}\subset S$, the matrix of active columns is augmented, i.e., $\mathbf{\Phi_1}=[\mathbf{Q}\ \boldsymbol{e}_{j_1}]$ and
the new residual is computed, requiring the inversion of
\[
\mathbf{\Phi_1}^T\mathbf{\Phi_1}= \begin{bmatrix}
\mathbf{I}_m & \mathbf{Q}^T\boldsymbol{e}_{j_1}\\
\boldsymbol{e}_{j_1}^T\mathbf{Q} & 1
\end{bmatrix}.
\]
Taking into account that $\mathbf{I}_m$ is invertible and $\beta=1-||\mathbf{Q}^T\boldsymbol{e}_{j_1} ||_2^2 > 1/2$ (inequality \eqref{in:usual_bound_byhalf} for $|S|=1$) and using the \textit{Matrix Inversion Lemma} in block form, we obtain:
\[
(\mathbf{\Phi_1}^T\mathbf{\Phi_1})^{-1}= \begin{bmatrix}
\mathbf{I}_m + \mathbf{Q}^T\boldsymbol{e}_{j_1}\boldsymbol{e}_{j_1}^T\mathbf{Q}/ \beta & -\mathbf{Q}^T\boldsymbol{e}_{j_1}/ \beta\\
-\boldsymbol{e}_{j_1}^T\mathbf{Q}/ \beta & 1/ \beta
\end{bmatrix}.
\]
After a few lines of elementary algebra, we take
\begin{align*}
\mathbf{\Phi_1}&(\mathbf{\Phi_1}^T\mathbf{\Phi_1})^{-1}\mathbf{\Phi_1}^T = \mathbf{QQ}^T+\mathbf{QQ}^T\boldsymbol{e}_{j_1}\boldsymbol{e}_{j_1}^T\mathbf{QQ}^T/ \beta - \\
-&\boldsymbol{e}_{j_1}\boldsymbol{e}_{j_1}^T\mathbf{QQ}^T/ \beta
-\mathbf{QQ}^T\boldsymbol{e}_{j_1}\boldsymbol{e}_{j_1}^T/ \beta+\boldsymbol{e}_{j_1}\boldsymbol{e}_{j_1}^T/ \beta.
\end{align*}
Hence, the new residual $\boldsymbol{r}^{(1)}=\boldsymbol{y}- \mathbf{\Phi_1}(\mathbf{\Phi_1}^T\mathbf{\Phi_1})^{-1}\mathbf{\Phi_1}^T \boldsymbol{y}$, can be recast as
\begin{align}
\boldsymbol{r}^{(1)}=&(\mathbf{I}_n-\mathbf{QQ}^T - \mathbf{QQ}^T\boldsymbol{e}_{j_1}\boldsymbol{e}_{j_1}^T\mathbf{QQ}^T/ \beta + \boldsymbol{e}_{j_1}\boldsymbol{e}_{j_1}^T\mathbf{QQ}^T/ \beta \nonumber \\
&+\mathbf{QQ}^T\boldsymbol{e}_{j_1}\boldsymbol{e}_{j_1}^T/ \beta - \boldsymbol{e}_{j_1}\boldsymbol{e}_{j_1}^T \ \beta)\boldsymbol{u}_0.
\label{eq:residualform1}
\end{align}
Relation \eqref{eq:residualform1} could be simplified using the decomposition for the outlier vector, i.e., $\boldsymbol{u}_0=u_{j_1}\boldsymbol{e}_{j_1} + F_{S\setminus S_1}(\boldsymbol{u}_0)$, where the second part is the vector which has the same elements as $\boldsymbol{u}_0$ over the set ${S\setminus S_1}$, besides the $j_1-$th coordinate which is equal to zero. Obviously, this is a vector, $s-1$ sparse at most and its support is a subset of the support of $\boldsymbol{u}_0$.
Thus, we have:
\begin{align}
\boldsymbol{r}^{(1)}&=(\mathbf{I}_n-\mathbf{QQ}^T - \mathbf{QQ}^T\boldsymbol{e}_{j_1}\boldsymbol{e}_{j_1}^T\mathbf{QQ}^T/ \beta + \nonumber\\
&+ \boldsymbol{e}_{j_1}\boldsymbol{e}_{j_1}^T\mathbf{QQ}^T/ \beta) F_{S\setminus S_1}(\boldsymbol{u}_0)= \nonumber\\
&=\boldsymbol{v}_1 -\mathbf{QQ}^T\boldsymbol{v}_1,
\label{eq:residualform2}
\end{align}
where
$\boldsymbol{v}_1 =  F_{S\setminus S_1}(\boldsymbol{u}_0) + \gamma_1\boldsymbol{e}_{j_1}$,
$\gamma_1= \boldsymbol{e}_{j_1}^T\mathbf{QQ}^TF_{S\setminus S_1}(\boldsymbol{u}_0) / \beta.$

At this point, we should note that $\supp(\boldsymbol{v}_1)=\supp(\boldsymbol{u}_0)=S$, while $\gamma_1 \neq u_{j_1}$. Following a similar rational, for the next step, we impose $|r^{(1)}(i)|>|r^{(1)}(j)|$ for all $i\in S\setminus S_1$ and $j\in S^c$. Hence, using lower and upper bounds leads to
\begin{align}
|&r^{(1)}(i)| =|  \langle \boldsymbol{r}^{(1)}, \boldsymbol{e}_i \rangle |= | \langle \boldsymbol{v}_1 - \mathbf{QQ}^{T}\boldsymbol{v}_1, \boldsymbol{e}_i \rangle  |\geq  \nonumber \\
&\geq |u_i | - |\boldsymbol{e}_i^T \mathbf{QQ}^{T}\boldsymbol{v}_1 | \geq \min_{l\in S}|u_l| - \delta_s^2 ||\boldsymbol{v}_1||_2,
\label{eq:guarantee_bound_left2}
\end{align}
where we used that $\langle \boldsymbol{e}_{i_1}, \boldsymbol{e}_i \rangle=0$ for $i\in S\setminus S_1$, and
\begin{align}
|r^{(1)}(j)| &=|  \langle \boldsymbol{r}^{(1)}, \boldsymbol{e}_j \rangle |= | \langle \boldsymbol{v}_1 -\mathbf{QQ}^T\boldsymbol{v}_1, \boldsymbol{e}_j \rangle  |= \nonumber \\
&= |\boldsymbol{e}_j^T\mathbf{QQ}^{T}\boldsymbol{v}_1| \leq \delta_s^2 ||\boldsymbol{v}_1||_2,
\label{eq:guarantee_bound_right2}
\end{align}
where we exploited the relationship $\langle \boldsymbol{v}_1,\boldsymbol{e}_j \rangle=0$, for every $j\in S^c$, as well as lemma \ref{lem:vec_bound_fin}.

Imposing $\min_{l\in S}|u_l| - \delta_s^2 ||\boldsymbol{v}_1||_2>  \delta_s^2 ||\boldsymbol{v}_1||_2$, leads equivalently to
\begin{equation}
\delta_s<\sqrt{\frac{\min_{l\in S}|u_l|}{2||\boldsymbol{v}_1||_2}}.
\label{eq:secondbound}
\end{equation}
Although \eqref{eq:secondbound}, seems inadequate, we will show indeed that this is a condition which always holds true, provided \eqref{eq:delta_bound_definition} is satisfied. One needs to prove that $||\boldsymbol{u}_0 ||_2>||\boldsymbol{v}_1||_2$, which is equivalent to showing that $|\gamma_1|<|u_{j_1}|$, using the aforementioned decompositions of $\boldsymbol{u}_0,\ \boldsymbol{v}_1$ and the Pythagorean Theorem.
Thus, we have that
\begin{align}
|\gamma_1|&= \frac{| \langle   \boldsymbol{e}_{j_1},\mathbf{QQ}^T F_{S\setminus S_1}(\boldsymbol{u}_0)  \rangle |}{| \beta|}
\leq 2\delta_s^2 ||F_{S\setminus S_1}(\boldsymbol{u}_0)||_2  \nonumber \\
& < \min_{j \in S} |u_j|\leq |u_{j_1}|,
\label{eq:lastbound}
\end{align}
due to $\beta>1/2$, the definition of the principal angle \eqref{EQ:principal_angle_final}, inequality \eqref{eq:delta_bound_definition} and  $||F_{S\setminus S_1}(\boldsymbol{u}_0)||_2 < ||\boldsymbol{u}_0 ||_2$, for any non-empty set $S_1$.

At the k-step $S_k=\{j_1,j_2,...,j_k \}\subset S$ and
the matrix that corresponds to the set of active columns is $\mathbf{\Phi_k}=[\mathbf{Q}\ \mathbf{I}_{S_k}]$. Using again the \textit{Matrix Inversion Lemma} for the inversion of $\mathbf{\Phi}_k^T\mathbf{\Phi}_k$, the new residual is given as follows:
\begin{align}
\boldsymbol{r}^{(k)}&=(\mathbf{I}_n-\mathbf{QQ}^T - \mathbf{QQ}^T\mathbf{I}_{S_k}\mathbf{M}_k^{-1}\mathbf{I}_{S_k}^T\mathbf{QQ}^T + \nonumber\\
&+ \mathbf{I}_{S_k}\mathbf{M}_k^{-1}\mathbf{I}_{S_k}^T\mathbf{QQ}^T)F_{S\setminus S_k}(\boldsymbol{u}_0) \nonumber\\
&=\boldsymbol{v}_k -\mathbf{QQ}^T\boldsymbol{v}_k,
\label{eq:residualformk}
\end{align}
where we used the identities
\begin{align}
\mathbf{M}_k &= \mathbf{I}_k -\mathbf{I}_{S_k}^T\mathbf{QQ}^T\mathbf{I}_{S_k}, \label{eq:Mk_matrix} \\
\boldsymbol{u}_0 &= F_{S_k}(\boldsymbol{u}_0)+ F_{S\setminus S_k}(\boldsymbol{u}_0), \label{eq:dec_u} \\
\boldsymbol{v}_k &= F_{S\setminus S_k}(\boldsymbol{u}_0)+  \mathbf{I}_{S_k}\mathbf{M}_k^{-1}\mathbf{I}_{S_k}^T\mathbf{QQ}^TF_{S\setminus S_k}(\boldsymbol{u}_0), \label{eq:dec_uk}
\end{align}
It is not hard to verify that $\supp(\boldsymbol{v}_k)=\supp(\boldsymbol{u}_0)=S$ still holds true. For a correct outlier index selection from the set $S$, at $k+1$ step, one needs to impose $|r^{(k)}(i)|>|r^{(k)}(j)|$ for all $i\in S\setminus S_k$ and $j\in S^c$. Using lower and upper bounds on the inner products, one obtains relations similar to \eqref{eq:guarantee_bound_left2}, \eqref{eq:guarantee_bound_right2} with $\boldsymbol{v}_k$ instead of $\boldsymbol{v}_1$, which leads to
\begin{equation}
\delta_s<\sqrt{\frac{\min_{l\in S}|u_l|}{2||\boldsymbol{v}_k||_2}}.
\label{eq:finalbound}
\end{equation}
The proof ends, by showing that the last bound is looser than that of inequality \eqref{eq:delta_bound_definition}, simply by proving that $||\boldsymbol{v}_k||_2< || \boldsymbol{u}_0||_2$ for all $k=1,2,...,s-1$.

Using the decompositions of these vectors \eqref{eq:dec_u}, \eqref{eq:dec_uk} and the Pythagorean Theorem, it suffices to show that
$|| \mathbf{M}_k^{-1}\mathbf{I}_{S_k}^T\mathbf{QQ}^T F_{S\setminus S_k}(\boldsymbol{u}_0)||_2 < || F_{S_k}(\boldsymbol{u}_0)||_2$, which follows from the fact that
\begin{align}
|| \mathbf{M}_k^{-1}&\mathbf{I}_{S_k}^T\mathbf{QQ}^T F_{S\setminus S_k}(\boldsymbol{u}_0)||_2 \nonumber\\
&\leq || \mathbf{M}_k^{-1}||_2 || \mathbf{I}_{S_k}^T\mathbf{QQ}^T F_{S\setminus S_k}(\boldsymbol{u}_0)||_2 \nonumber\\
&<  \min_{l \in S}|u_l|\leq || F_{S_k}(\boldsymbol{u}_0)||_2,
\end{align}
where we employed the sub-multiplicative property of the matrix 2-norm, inequality \eqref{eq:norm<2}, Lemma \ref{lem:vec_bound_fin} and \eqref{eq:delta_bound_definition}.

Finally, at step $k+1=s$ it is guaranteed that the correct support is recovered and thus the linear subspace, onto which the measurement vector $\boldsymbol{y}$ lies, is build. In turn, this results to a Least Squares solution for GARD of zero error.
\end{proof}

\section{Proof of Theorem \ref{theor:exact_support}}
\label{appendix:B}
Since, theorem \ref{theor:exact_support} is the generalization of \ref{theor:exact_rec}, some intermediate results regarding the proof presented in Appendix \ref{appendix:A}, will also be used here. On the other hand, we will try to avoid the technical parts with shared similarities. 
\begin{proof}
Due to the existence and uniqueness of the \textit{QR} decomposition, the analysis is based on equation \eqref{eq:probdefmat_qr_noise}.\\
Since initially GARD performs a Least Squares step, where the columns that participate in the representation are only those of matrix $\mathbf{X}$, the residual is also $\boldsymbol{r}^{(0)}=(\mathbf{I}_n - \mathbf{QQ}^T)\boldsymbol{y}$. Thus, taking into account \eqref{eq:probdefmat_qr_noise}, we have the following expression for the initial residual:
\[
\boldsymbol{r}^{(0)}= \boldsymbol{u}_0 + \boldsymbol{\eta} - 
\mathbf{QQ}^T\boldsymbol{u}_0 - \mathbf{QQ}^T\boldsymbol{\eta},
\]
where the extra terms are due to the noise vector $\boldsymbol{\eta}$.\\
Once again, we should impose \eqref{eq:guarantee}, according to \eqref{eq:guarantee_bound_left} and \eqref{eq:guarantee_bound_right}. Also, recall on Theorem \ref{theor:exact_rec}, suggesting\footnote{In the noiseless case, $\sqrt{2}/2$ was the upper bound for $c>\delta_s$, achieved only for $1$-sparse outlier vectors. Thus, if $\delta_s$ exceeds this limit, GARD has little chance in recovering the correct support, even in the presence of outlier noise only, let alone as inlier noise coexists.} that $\delta_s<c\leq \sqrt{2}/2.$ Thus, we have:   
\begin{align*}
|r^{(0)}(i)| &\geq |u_i | - |\langle \mathbf{QQ}^{T}\boldsymbol{u}_0, \boldsymbol{e}_i \rangle | - | \langle \boldsymbol{\eta}, \boldsymbol{e}_i \rangle | - 
|\langle \mathbf{QQ}^{T}\boldsymbol{\eta}, \boldsymbol{e}_i \rangle | \\ 
&\geq \min_{l\in S}|u_l| - \delta_s^2 ||\boldsymbol{u}_0||_2 -\epsilon_0 - \epsilon_0\delta_s \nonumber \\
& > \min_{l\in S}|u_l| - \delta_s^2 ||\boldsymbol{u}_0||_2 -\epsilon_0 - \frac{\epsilon_0}{\sqrt{2}}\\
& > \min_{l\in S}|u_l| - \delta_s^2 ||\boldsymbol{u}_0||_2 -\epsilon_0 - \epsilon_0 \sqrt{\frac{3}{2}}
\end{align*}
and
\begin{align*}
|r^{(0)}(j)| & \leq \epsilon_0 + \delta_s^2 ||\boldsymbol{u}_0||_2+ \epsilon_0\delta_s\\
&< \epsilon_0 + \delta_s^2 ||\boldsymbol{u}_0||_2+  \frac{\epsilon_0}{\sqrt{2}}\\
&< \epsilon_0 + \delta_s^2 ||\boldsymbol{u}_0||_2+  \epsilon_0 \sqrt{\frac{3}{2}},
\end{align*}
for $i \in S$ and $j \in S^c$, respectively. Thus, inequality \eqref{eq:delta_bound_definition_epsilon} follows for the initial step. in the following, we proceed with the general selection step at $k+1$, as the first one is omitted, since it could be viewed as a special case of the general $k+1$ step. In the proof of Theorem \ref{theor:exact_rec}, it was presented for comprehension reasons solely. It should also be noted, that the matrices augmented and inverted at each step, are those presented in the proof of Theorem \ref{theor:exact_rec}. However, this is not the case for the solution and the residual, which is in our greatest interest.\\
The condition in \eqref{eq:delta_bound_definition_epsilon}, guarantees, that at each selection step the support of our sparse estimate is a subset of the sparse outlier vector $\boldsymbol{u}_0$, i.e., $S_k \subset S$ and the matrix that corresponds to the set of active columns is $\mathbf{\Phi}_k = [\mathbf{Q}\ \mathbf{I}_{S_k} ]$. Employing familiar techniques, we have an expression for the residual at the $k$ step:
\begin{equation}
\label{eq:rk_sup_rec}
\boldsymbol{r}^{(k)}= \boldsymbol{v}_k + \boldsymbol{\eta}_k - 
\mathbf{QQ}^T\boldsymbol{v}_k - \mathbf{QQ}^T\boldsymbol{\eta}_k,
\end{equation}
where $\boldsymbol{v}_k$ is the vector defined in \eqref{eq:dec_uk} and 
\begin{equation}
\label{eq:eta_k}
\boldsymbol{\eta}_k = F_{J\setminus S_k}(\boldsymbol{\eta})+  \mathbf{I}_{S_k}\mathbf{M}_k^{-1}\mathbf{I}_{S_k}^T\mathbf{QQ}^TF_{J\setminus S_k}(\boldsymbol{\eta}).
\end{equation}
In \eqref{eq:eta_k}, it is clear that the only differences between $\boldsymbol{\eta}_k$ and $\boldsymbol{\eta}$,
take place at the elements indexed $j_k \in S_k$, i.e, indices that GARD has selected as outliers. Moreover, for $J'=J \setminus S_k$ holds $J'\cap S_k= \emptyset,$ i.e., the vector is decomposed into two disjoint subsets. At this point, prior to completing the proof, it is required to establish appropriate bounds for the inner products $\left| \langle \boldsymbol{e}_i, \mathbf{QQ}^T \boldsymbol{\eta}_k   \rangle \right|$ and $\left|\langle \boldsymbol{e}_i,\boldsymbol{\eta}_k \rangle \right|$. Due to the Pythagorean Theorem, \eqref{eq:norm_bound} and \eqref{eq:norm<2}
\begin{align*}
\left\| \boldsymbol{\eta}_k \right\|_2^2 &= \left\| F_{J\setminus S_k}(\boldsymbol{\eta}) \right\|_2^2 + \left\| \mathbf{M}_k^{-1} \mathbf{I}_{S_k}^T \mathbf{QQ}^T F_{J \setminus S_k}(\boldsymbol{\eta}) \right\|_2^2  \\
& \leq \epsilon_0^2+ 2\epsilon_0^2=3\epsilon_0^2.
\end{align*} 
Hence,
\[
\left| \boldsymbol{e}_i^T\mathbf{QQ}^T \boldsymbol{\eta}_k \right| \leq \delta_s \left\| \boldsymbol{\eta}_k \right\|_2 \leq \delta_s \sqrt{3} \epsilon_0 < \epsilon_0 \sqrt{\frac{3}{2}},
\]
where we have also used the maximum bound, that $\delta_s<\sqrt{2}/2$. Also, for all $i \in J \setminus S_k$, holds $\left|\langle \boldsymbol{e}_i,\boldsymbol{\eta}_k \rangle \right| \leq \left|\langle \boldsymbol{e}_i,F_{J\setminus S_k}(\boldsymbol{\eta}_k) \rangle \right| \leq \epsilon_0$.
Thus, adopting bounds to the absolute value of the inner products, we have
\begin{align*}
|r^{(k)}(i)| &\geq |u_i | - |\langle \mathbf{QQ}^{T}\boldsymbol{v}_k, \boldsymbol{e}_i \rangle |- \\ & \quad - | \langle \boldsymbol{\eta}_k, \boldsymbol{e}_i \rangle | - 
|\langle \mathbf{QQ}^{T}\boldsymbol{\eta}_k, \boldsymbol{e}_i \rangle | \geq \\ 
&\geq \min_{l\in S}|u_l| - \delta_s^2 ||\boldsymbol{v}_k||_2 -\epsilon_0 - \epsilon_0 \sqrt{\frac{3}{2}}
\end{align*}
and
\[
|r^{(k)}(j)| \leq  \delta_s^2 ||\boldsymbol{v}_k||_2+\epsilon_0 + \epsilon_0 \sqrt{\frac{3}{2}},
\]
for $i \in S \setminus S_k$ and $j \in S^c$, respectively. 
Thus, imposing $|r^{(k)}(i)|>|r^{(k)}(j)|$, leads to 
\[
\delta_s< \sqrt{\frac{\min_{l\in S}|u_l|-(2+\sqrt{6})\epsilon_0}{2 ||\boldsymbol{v}_k||_2}},
\]
which is satisfied, suppose \eqref{eq:delta_bound_definition_epsilon} holds true. This  holds true, due to the fact that $\left\| \boldsymbol{v}_k \right\|_2< \left\| \boldsymbol{u}_o \right\|_2$ for all $k=1,2,...,s-1$ (read at the end of Appendix \ref{appendix:A} for the proof).
\end{proof}



\bibliographystyle{IEEEtran}
\bibliography{Geo_Library}

\end{document}